\documentclass[runningheads]{llncs}
\usepackage{graphicx}
\usepackage{lipsum}

\usepackage{graphicx}%
\usepackage{multirow}%
\usepackage{amsmath,amssymb,amsfonts}%
\usepackage{mathrsfs}%
\usepackage[title]{appendix}%
\usepackage{xcolor}%
\usepackage{textcomp}%
\usepackage{manyfoot}%
\usepackage{booktabs}%
\usepackage{algorithm}%
\usepackage{algorithmicx}%
\usepackage{algpseudocode}%
\usepackage{listings}%
\usepackage{adjustbox}

\usepackage{cite}

\let\proglang=\textsf

\usepackage{graphicx}
\usepackage{float}
\graphicspath{ {images/} }
\usepackage{multirow}
\usepackage{subcaption}

\usepackage{listings}

\newcommand{\X}{\mathrm{X}}
\newcommand{\Y}{\mathrm{Y}}

\newcommand{\dcor}{\mathcal{R}}
\newcommand{\R}{\mathbb{R}}
\newcommand{\dcov}{\mathcal{V}}
\begin{document}
	\title{Improved distance correlation estimation}
	%
	%
	\author{Blanca E. Monroy-Castillo\inst{1} \and
		M.A. Jácome \inst{2} \and
		Ricardo Cao \inst{1}}
	\authorrunning{Monroy-Castillo et al.}
	%
	\institute{Research group MODES, CITIC, Department of Mathematics, Universidade da Coru\~na, A Coru\~na, Spain\\
		\email{b.mcastillo@udc.es}\\
		\and
		Research group MODES, CITIC, Department of Mathematics, Faculty of Sciences, Universidade da Coru\~na, A Coru\~na, Spain\\
		}
	\maketitle              
	\begin{abstract}
		Distance correlation is a novel class of multivariate dependence measure, taking positive values between 0 and 1, and applicable to random vectors of arbitrary dimensions, not necessarily equal. It offers several advantages over the well-known Pearson correlation coefficient, the most important is that distance correlation equals zero if and only if the random vectors are independent. 
		There are two different estimators of the distance correlation available in the literature. The first one, proposed by \cite{szekely_measuring_2007}, is based on an asymptotically unbiased estimator of the distance covariance which turns out to be a V-statistic. The second one builds on an unbiased estimator of the distance covariance proposed in \cite{szekely_partial_2014}, proved to be an U-statistic by \cite{huo_fast_2016}. This study evaluates their efficiency (mean squared error) and compares computational times for both methods under different dependence structures.	Under conditions of independence or near-independence, the V-estimates are biased, while the U-estimator frequently cannot be computed due to negative values. To address this challenge, a convex linear combination of the former estimators is proposed and studied, yielding good results regardless of the level of dependence.

	\end{abstract}

	\section{Introduction}\label{Sec:Introduction}
	
	The concept of dependence among random observations plays a central role in many fields, including statistics, medicine, biology and engineering, among others. Given the inherent complexity of fully understanding dependencies, the strength of these relationships is often distilled into a single metric, the correlation coefficient. Numerous types of correlation coefficients exist, but perhaps the most widely known is Pearson correlation coefficient. Additionally, alternative measures of correlation exist. For instance, rank correlation assesses the relationship between the rankings of two variables or the rankings of the same variable across different conditions. Examples of rank correlation coefficients include Spearman's rank correlation coefficient, Kendall tau correlation coefficient, and Goodman and Kruskal's gamma.\\
	
	Pearson correlation coefficient presents some disadvantages. Firstly, the two variables must be normally distributed. Secondly, the fact that the Pearson correlation coefficient is zero does not determine independence between two variables, as only a linear dependence between the variables can be determined and the variables may have a nonlinear relationship. In recent years, a novel measure of dependence between random vectors has been proposed, \textit{distance correlation}, was introduced by \cite{szekely_measuring_2007}. They point out that distance covariance and distance correlation share a parallel with product-moment covariance and correlation. However, unlike the classical definition of correlation, distance correlation is zero solely when the random vectors are independent. Besides, the distance correlation can be used to evaluate both linear and nonlinear correlations between variables. Essentially, for all distributions with finite first moments, distance correlation ($\dcor$) extends the notion of correlation in two essential ways:
	
	\begin{enumerate}
		\item $\dcor(X,Y)$ is defined for $X$ and $Y$ in arbitrary dimensions;
		\item $\dcor(X,Y) = 0$ characterizes independence of $X$ and $Y$.  
	\end{enumerate}
	Distance correlation satisfies $0 \leq \dcor \leq 1$, and $\dcor = 0$ if and only if $X$ and $Y$ are independent. \\
	
	Distance correlation has been applied and extended to a great variety of fields, such as variable selection \cite{yenigun_variable_2015, febrero-bande_variable_2019, yang_sufficient_2019, wu_mm_2021}, as well as in disciplines like biology \cite{SunHerazoMayaHuangKaminskiZhao+2018, brankovic_distributed_2018} and medicine \cite{genes9120608, Wenxing}, among others. Moreover, it has been explored within high-dimensional contexts \cite{yao_testing_2018, lu_conditional_2021}. Furthermore, the applicability of distance correlation has been expanded to address the challenge of testing independence of high-dimensional random vectors in \cite{szekely_distance_2013}. In a similar way, the concept of partial distance correlation was introduced by \cite{szekely_partial_2014}.
	Distance correlation has undergone an extension to encompass conditional distance correlation \cite{wang_conditional_2015, lu_model-free_2020, lu_conditional_2021, cui_model-free_2022}, and it has also been examined in the realm of survival data \cite{chen_feature_2021, chen_robust_2018, chen_ultrahigh-dimensional_2022, edelmann_consistent_2022, zhang_model-free_2021}.
	Lastly, \cite{edelmann_relationships_2021} proves that for any fixed Pearson correlation coefficient strictly between -1 and 1, the distance correlation coefficient can attain any value within the open unit interval (0,1).\\
	
	The estimation of the distance correlation relies on the estimation of the distance covariance. \cite{szekely_uniqueness_2012} demonstrated the uniqueness of distance covariance. \cite{szekely_measuring_2007} proposed a sample distance covariance estimator, and they proved that it turns out to be a V-statistic. This estimator of the distance covariance leads to the so-called V-estimator of the distance correlation. Moreover, intermediate results in \cite{szekely_partial_2014} led to an unbiased estimator of the squared distance covariance. This unbiased estimator was further identified as a U-statistic in \cite{huo_fast_2016}. This distance covariance estimator results in the U-estimator of the distance correlation. \\
	
	The two estimators of the distance correlation, U-estimator and V-estimator, offer different properties. The direct implementation of the V-estimator results in computational complexity that scales as $\mathcal{O}(n^2)$. The computation of the U-estimator reduces computational complexity to $\mathcal{O}(n \log n)$. But both the U-statistic and the V-statistic of distance correlation can be calculated in $\mathcal{O}(n \log n)$ steps, the algorithm in \cite{huo_fast_2016} can be straightforward extended to the V-statistic version. Both estimators exhibit good asymptotic properties; however, the U-estimator allows for easier derivation of these properties.\\
	
	One scenario of interest when working with distance correlation is that of independence, when both distance correlation and distance covariance are zero. It is under independence or situations close to it where the two estimators exhibit the most significant differences. Firstly, the U-estimator of the squared distance covariance is unbiased, leading to occurrences of negative values of the squared distance covariance, which precludes the computation of the U-estimator of the distance correlation. Conversely, the V-estimator biased as it can only take positive values. However, to the best of our knowledge, the respective advantages and disadvantages of each distance correlation estimator, and the use of one or the other estimator, do not seem to be based on any specific basis. For example, while \cite{edelmann_consistent_2022} employ the U-estimator to propose an extension for right-censored data, \cite{wang_conditional_2015} suggest an extension for conditional distance correlation using both estimators.\\
	
	In this work, a simulation study to assess the practical behavior and efficiency of each estimator with different dependency models is conducted. The experimental results show that neither is consistently better, making the choice of the estimator in practice a challenging task. To tackle this inconvenience a new approach, a convex linear combination of the former estimators, is introduced and studied. \\
	
	The remainder of this paper is organized as follows. Section \ref{sec:preliminaries} introduces the preliminaries, defining distance covariance and distance correlation. Additionally, it presents the estimators of the distance covariance and distance correlation in the literature. Some of the existing packages developed in R and Python software are introduced. Upon studying the occurrence of negative values for the squared distance covariance estimator which makes it impossible to calculate the U-estimator of the distance correlation, two modifications of the U-estimator are proposed to handle that problem.  Section \ref{sec:convex} introduces a convex linear combination approach to address the estimator choice problem. Section \ref{sec:simulation} shows the results of the simulation study through three models: Farlie-Gumbel-Morgenstern (FGM), a bivariate normal, and a nonlinear model. The results are compared and shown in terms of efficiency, mean squared error (MSE), bias, variance, and computational time for each estimator, moreover,  a comprehensive comparison between the original estimators and the proposed alternatives are offered. Finally, Section \ref{sec:Conclusions} provides the concluding remarks. 
	
	\section{Distance correlation estimation} \label{sec:preliminaries}
	
	Let $X \in \R^p$ and $Y \in \R^q$ be random vectors, where $p$ and $q$ are positive integers. The characteristic function of $X$ and $Y$ are denoted $\phi_X$ and $\phi_Y$, respectively, and the joint characteristic function of $X$ and $Y$ is denoted $\phi_{X,Y}$. \\
	\begin{definition}
		The distance covariance between random vectors $X$ and $Y$ with finite first moments is the nonnegative square root of the number $\dcov^2(X, Y)$, defined by:
		\begin{eqnarray*}
			\dcov^2(X,Y) &=& \lVert \phi_{X,Y}(t,s)-\phi_X(t)\phi_Y(s) \rVert^2 \\
			&=& \frac{1}{c_pc_q}\int_{\R^{p+q}} \frac{\lvert \phi_{X,Y}(t,s)-\phi_X(t)\phi_Y(s)\rvert^2}{|t|^{1+p}_p|s|^{1+q}_q}dt\,ds,
		\end{eqnarray*}
		where $c_d = \frac{\pi^{(1+d)/2}}{\Gamma((1+d)/2)}$. Similarly, the distance variance is defined as the square root of 
		\begin{eqnarray*}
			\dcov^2(X) = \dcov^2(X,X) = \lVert \phi_{X,X}(t,s)-\phi_X(t)\phi_X(s) \rVert^2.
		\end{eqnarray*}
	\end{definition}
	
	\begin{definition}
		The distance correlation between two random vectors $X$ and $Y$ with finite first moments is the positive square root of the nonnegative number $\dcor^2(X,Y)$ defined by 
		\begin{eqnarray}\label{eq:dcor}
			\dcor^2(X,Y) &=& \begin{cases}
				\frac{\dcov^2(X,Y)}{\sqrt{\dcov^2(X)\dcov^2(Y)}}, & \dcov^2(X)\dcov^2(Y) > 0 \\
				0, & \dcov^2(X)\dcov^2(Y) = 0.
			\end{cases}
		\end{eqnarray}
	\end{definition}
	
	On the other hand, it is showed an equivalent form to compute the distance covariance through expectations given by \cite{szekely_measuring_2007}, this is, if $E|X|^2_p < \infty$ and $E|Y|^2_q < \infty$, then $E[|X|_p|Y|_q] < \infty$, and
	\begin{eqnarray}\label{eq:expectations}
		\dcov^2(X,Y) &=& E[|X_1-X_2|_p|Y_1-Y_2|_q] + E[|X_1-X_2|_p]E[|Y_1-Y_2|_q] \nonumber \\ & & \, - 2E[|X_1-X_2|_p|Y_1-Y_3|_q],
	\end{eqnarray}
	where $(X_1,Y_1),(X_2,Y_2)$ and $(X_3,Y_3)$ are independent and identically distributed as $(X,Y)$. Note that with Equation (\ref{eq:expectations}), it is possible to compute the distance covariance using only the density function, without needing to know the characteristic function. In the same way, it is possible to compute the distance variance of $X$ as the square root of 
	\begin{eqnarray}\label{eq:expectationsvar}
		\dcov^2(X,X)  \,\,\,=\,\,\,  \dcov^2(X) &=& E\left[|X_1-X_2|_p^2\right] + E\left[|X_1-X_2|_p\right]^2  \nonumber \\ & & \, - 2E\left[|X_1-X_2|_p|X_1-X_3|_p\right],
	\end{eqnarray}
	similarly for $\dcov^2$($Y$).
	As a result, it becomes possible to accurately determine the exact distance correlation $\dcor(X,Y)$ using (\ref{eq:expectations}), (\ref{eq:expectationsvar}) and (\ref{eq:dcor}).\\ 
	
	For an observed random sample $(\boldsymbol{\X,\Y}) = \{(X_k,Y_k): k = 1, \dots,n\}$ from the joint distribution of random vectors $X \in \R^p$ and $Y \in \R^q$, \cite{szekely_measuring_2007} proposed the following empirical estimator of the distance covariance. \\
	\begin{definition}
		The empirical distance covariance $\dcov_n(\boldsymbol{\X}, \boldsymbol{\Y})$ is defined by the nonnegative square root of:
		\begin{eqnarray}\label{eq:V-statistic}
			\dcov_n^2(\boldsymbol{\X}, \boldsymbol{\Y}) &=& \frac{1}{n^2} \sum_{k,l = 1}^n A_{kl}B_{kl},
		\end{eqnarray}
		where $A_{kl}$ and $B_{kl}$ denote the corresponding double-centered distance matrices defined as 
		\begin{eqnarray*}
			A_{kl} &=& \begin{cases}
				a_{kl} - \frac{1}{n} \sum_{j = 1}^n a_{kj} - \frac{1}{n} \sum_{i = 1}^n a_{il} + \frac{1}{n^2} \sum_{i,j=1}^n a_{ij}, & k \neq l \\
				0, & k = l,
			\end{cases}
		\end{eqnarray*}
		with $a_{kl} = |X_k - X_l|$ are the pairwise distances of the $X$ observations. The terms $B_{kl}$ are defined similarly but using $b_{kl} = |Y_k-Y_l|$ instead of $a_{kl}$.  In the same way the sample distance variance $\dcov_n(\boldsymbol{\X})$ is defined as the square root of: 
		\begin{eqnarray*}
			\dcov_n^2(\boldsymbol{\X}) &=& \dcov_n^2(\boldsymbol{\X,\X}) \,\,\, = \,\,\, \frac{1}{n^2}\sum_{k,l = 1}^n A_{kl}^2.
		\end{eqnarray*}
	\end{definition}
	Theorem 1 in \cite{szekely_measuring_2007} proved that $\dcov_n^2(\boldsymbol{\X,\Y}) \geq 0$. Moreover, it is proved that under independence, $\dcov_n^2(\boldsymbol{\X,\Y})$ is a degenerate kernel V-statistic. \\
	
	The first estimator of the distance correlation in \cite{szekely_measuring_2007}, dCorV, is based on the empirical distance variance and covariance, resulting the empirical distance correlation. \\
	\begin{definition}
		The empirical distance correlation $\dcor_n(\boldsymbol{\X}, \boldsymbol{\Y})$ is the square root of
		\begin{eqnarray}\label{eq:dcorV}
			\mathrm{dCorV}^2(\boldsymbol{\X,\Y}) \,\,\, = \,\,\, \dcor_{n}^2(\boldsymbol{\X}, \boldsymbol{\Y}) &=& \begin{cases}
				\frac{\dcov_n^2(\boldsymbol{\X,\Y})}{\sqrt{\dcov_n^2(\boldsymbol{\X})\dcov_n^2(\boldsymbol{\Y})}}, & \dcov_n^2(\boldsymbol{\X})\dcov_n^2(\boldsymbol{\Y}) > 0 \\ 0, & \dcov_n^2(\boldsymbol{\X})\dcov_n^2(\boldsymbol{\Y}) = 0.
			\end{cases}
		\end{eqnarray}
	\end{definition}
	
	Likewise, \cite{szekely_partial_2014} proposed an alternative estimator for the distance covariance $\dcov(X,Y)$ based on a $\mathcal{U}-$centered matrix. \\
	\begin{definition}
		Let $A = (a_{kl})$ be a symmetric, real valued $n \times n$ matrix with zero diagonal, $n > 2$. Define the $\mathcal{U}-$centered matrix $\tilde{A}$ as follows. Let the $(k,l)$th entry of $\tilde{A}$ be defined by 
		\begin{eqnarray*}
			\tilde{A}_{kl} = \begin{cases}
				a_{kl} - \frac{1}{n-2} \sum_{j=1}^n a_{kj}-\frac{1}{n-2} \sum_{i = 1}^n a_{il} + \frac{1}{(n-1)(n-2)}\sum_{i,j=1}^n a_{ij}, & k \neq l; \\ 0, & k = l.
			\end{cases}
		\end{eqnarray*}
	\end{definition}
	Here "$\mathcal{U}-$centered" is so named because the inner product, 
	\begin{eqnarray} \label{eq:U-statistic}
		\mathcal{U}_n^2(\boldsymbol{\X,\Y}) \,\,\, = \,\,\, \left(\tilde{A}\cdot \tilde{B}\right) &=& \frac{1}{n-3} \sum_{k \neq l} \tilde{A}_{kl}\tilde{B}_{kl},
	\end{eqnarray}
	defines an unbiased estimator of the squared distance covariance $\dcov^2(X,Y)$. \cite{huo_fast_2016} established that the estimator in Equation (\ref{eq:U-statistic}) can be expressed as a U-statistic. Thus, it becomes feasible to define an alternative estimator of the empirical distance correlation through $\mathcal{U}_n(\boldsymbol{\X,\Y})$ and distance variance of $\boldsymbol{\X}$ and $\boldsymbol{\Y}$, $\mathcal{U}_n(\boldsymbol{\X})$, $\mathcal{U}_n(\boldsymbol{\Y})$, respectively, in the following manner.\\
	\begin{definition}
		The estimator dCorU of the distance correlation $\dcor(X,Y)$, based on U-statistics, is the square root of
		\begin{eqnarray}\label{eq:dcorU}
			\mathrm{dCorU}^2(\boldsymbol{\X,\Y}) &=& \begin{cases}
				\frac{\mathcal{U}_n^2(\boldsymbol{\X,\Y})}{\sqrt{\mathcal{U}_n^2(\boldsymbol{\X})\mathcal{U}_n^2(\boldsymbol{\Y})}}, & \mathcal{U}_n^2(\boldsymbol{\X})\mathcal{U}_n^2(\boldsymbol{\Y}) > 0 \\ 0, & \mathcal{U}_n^2(\boldsymbol{\X})\mathcal{U}_n^2(\boldsymbol{\Y}) = 0,
			\end{cases}
		\end{eqnarray}
	\end{definition}
	This reformulation facilitates a fast algorithm for estimating the distance covariance, which can be implemented with a computational complexity of $\mathcal{O}(n \log n)$, while the original estimator (Eq. (\ref{eq:V-statistic})) has a computational complexity of $\mathcal{O}(n^2)$. Alternatively, \cite{chaudhuri_fast_2019} proposed an algorithm primarily composed of two sorting steps for computing the estimator of the distance correlation in (\ref{eq:dcorU}). This design renders it simple to implement and also results in a computational complexity of $\mathcal{O}(n \log n)$. \\
	
	These results have prompted the implementation and development of many software packages, available in R \cite{R} and Python \cite{van1995python}. A comprehensive comparison of the performance between these packages in both languages is presented by \cite{RAMOS2023101326}. An open-source Python package for distance correlation and other statistics is introduced: the \textbf{\proglang{dcor}} package \cite{dcorpython}. The studied libraries in Python are \textbf{\proglang{statsmodels}} \cite{seabold2010statsmodels}, \textbf{\proglang{hyppo}} \cite{hyppo}, and \textbf{\proglang{pingouin}} \cite{Vallat2018}. In R, the \textbf{\proglang{energy}} \cite{energy} package implements the \proglang{dcov} and \proglang{dcor} functions, which return $\dcov_n(\boldsymbol{\X,\Y})$ and dCorV($\boldsymbol{\X,\Y}$), respectively. For the U-estimator, the \proglang{dcovU} and \proglang{bcdcor} functions are implemented. These functions return the $\mathcal{U}_n^2$($\boldsymbol{\X,\Y}$) and dCorU$^2$($\boldsymbol{\X,\Y}$) values, respectively; that is, they do not take the square root. On the other hand, \textbf{\proglang{dcortools}} implements the \proglang{distcov} and \proglang{distcor} functions. The argument \proglang{bias.corr} allows to use the V-estimator when the argument is \proglang{FALSE}, or the U-estimator with \proglang{bias.corr = TRUE}.
	Additionally, the \textbf{\proglang{Rfast}} \cite{Rfast} package includes functions as \proglang{dvar}, \proglang{dcov}, \proglang{dcor}, and \proglang{bcdcor}, using the fast method proposed by \cite{huo_fast_2016}. And the \proglang{bcdcor} function computes the bias-corrected distance correlation of two matrices. \\
	
	A notable aspect to consider, as previously mentioned, is that the computation of  $\mathcal{U}_n^2(\boldsymbol{\X,\Y})$ in Equation (\ref{eq:U-statistic}) can yield negative values in cases of independence or very low levels of dependence and small sample sizes. Consequently, it becomes impossible to calculate dCorU using the expression in Equation (\ref{eq:dcorU}). This precludes the computation of the dCorU estimator as the square root of dCorU$^2$. This issue arises from the fact that $\mathcal{U}_n^2(\boldsymbol{\X,\Y})$ is an unbiased estimator of the squared distance covariance $\dcov^2(X,Y)$, which is 0 under independence. To the best of our knowledge, this problem has not been discussed in the literature. However some authors have implemented this estimator in software such as R and Python, giving the following alternatives to obtain the computation of dCorU. Authors as \cite{energy,Rfast} return dCorU$^2$ without computing the square root, while others, as \cite{dcortools} return for dCorU the square root of the absolute value of dCorU$^2$, with the sign corresponding to dCorU$^2$. However, as shown in Table \ref{tab:negatives} in Section \ref{sec:negatives}, when working in scenarios under independence, approximately 60\% of the estimates are negative values. For this reason, this study explores two proposals to address this challenge, resulting in the following alternative estimators of the distance correlation:
	
	\begin{itemize}
		\item Replace the values of $\mathcal{U}_n^2(X,Y)$ with their absolute value. This U-estimator of the distance correlation is denoted by dCorU(A).
		\item Consider $\max(\mathcal{U}_n^2(X,Y), 0)$ so negative values of $\mathcal{U}_n^2(X,Y)$ are truncated to zero. Denote this U-estimator of $\dcor$ as dCorU(T).   
	\end{itemize}  
	
	Under an independence scenario, both dCorU and dCorU(A) yield identical MSE. This is because the squared of the difference between the estimated value and the real value ($(\text{dCorU} - \dcor)^2)$ does not deviate from the difference obtained when using the absolute value ($(\text{dCorU(A)} - \dcor)^2$) when $\dcor = 0$.
	
	\section{New proposal for the estimation of the distance correlation} \label{sec:convex}
	
	Under independence or near-independence conditions ($\dcor \approx 0$), the V-estimator is biased as it always provides positive results. On the other side, the U-estimator frequently cannot be computed because of negative values of dCorU$^2(X,Y)$. The preference between the two estimators seems to depend also on the nature (linear, nonlinear) of the relationship between $X$ and $Y$, see results in Section 4.1. While simulations allow for studying the behavior of estimators and gaining insight into when to use each, in real-world scenarios, determining whether the relationship between $X$ and $Y$ is linear or nonlinear, and assessing the level of dependence or independence, can be challenging. Consequently, selecting the appropriate estimator becomes difficult. To overcome this limitation, this paper introduces a new estimator of teh distance correlation that overcomes this challenges. It is a convex linear combination of both estimators (dCorU, dCorV), denoted as $$\text{dCor}_{\lambda} = \lambda\text{dCorU} + (1 - \lambda)\text{dCorV},$$ where $\lambda \in [0,1]$ serves as a weighting parameter determining the balance between the two estimators. It is proposed to use the optimal value $\lambda_0$ for the parameter $\lambda$, which minimizes the Mean Squared Error (MSE) of this new estimator. This optimal value $\lambda_0$, as provided in Lemma \ref{le:lambda}, relies on various unknown quantities such as the covariance, variance, and bias of dCorU and dCorV. \\
	
	\begin{lemma}\label{le:lambda}
		Let $X \in \mathbb{R}^p$ and $Y \in \mathbb{R}^q$ be random vectors, where $p$ and $q$ are positive integers and define $\hat{\theta}^U = $dCorU(X,Y) and $\hat{\theta}^V = $dCorV(X,Y). Given the convex linear combination $\lambda \hat{\theta}^U + (1-\lambda)\hat{\theta}^V$, then, in the sense of minimizing the MSE, the optimal value of $\lambda$ is given by $\lambda_0$ defined as: 
		
		\begin{small}
			\begin{align}\label{eq:lambda_opt}
				\lambda_0 &= \frac{-\text{Cov}\left(\hat{\theta}^U, \hat{\theta}^V\right) + \text{Var}\left(\hat{\theta}^V\right) + \text{Bias}\left(\hat{\theta}^V\right)\left(\text{Bias}\left(\hat{\theta}^V\right)-\text{Bias}\left(\hat{\theta}^U\right)\right)}{\text{Var}\left(\hat{\theta}^U\right) + \text{Var}\left(\hat{\theta}^V\right) - 2 \text{Cov}\left(\hat{\theta}^U, \hat{\theta}^V\right) + \left(\text{Bias}\left(\hat{\theta}^U\right) - \text{Bias}\left(\hat{\theta}^V\right)\right)^2}.
			\end{align}
		\end{small}
	\end{lemma}
	
	\begin{proof}
		Let $\hat{\theta}^C_{\lambda_0} = \lambda_0 \hat{\theta}^U + (1-\lambda_0)\hat{\theta}^V$ be the convex linear combination, then $$\text{MSE}\left(\hat{\theta}^C_{\lambda_0}\right) = \text{Var}\left(\lambda_0 \hat{\theta}^U + (1-\lambda_0)\hat{\theta}^V\right) + \text{Bias}\left(\lambda_0 \hat{\theta}^U + (1-\lambda_0)\hat{\theta}^V\right)^2,$$ where $$\text{Var}\left(\hat{\theta}^C_{\lambda_0}\right) =  \lambda_0^2\text{Var}\left(\hat{\theta}^U\right) + (1-\lambda_0)^2\text{Var}\left(\hat{\theta}^V\right) + 2\lambda_0(1-\lambda_0)\text{Cov}\left(\hat{\theta}^U, \hat{\theta}^V\right)$$ and $$\text{Bias}\left(\hat{\theta}^C_{\lambda_0}\right) = \lambda_0\text{Bias}\left(\hat{\theta}^U\right) + (1-\lambda_0)\text{Bias}\left(\hat{\theta}^V\right).$$ Therefore,
		\begin{eqnarray*}
			\text{MSE}\left(\hat{\theta}^C_{\lambda_0}\right) &=& \lambda_0^2\text{Var}\left(\hat{\theta}^U\right) + (1-\lambda_0)^2\text{Var}\left(\hat{\theta}^V\right) + 2\lambda_0(1-\lambda_0)\text{Cov}\left(\hat{\theta}^U, \hat{\theta}^V\right)  \\ & &\,\, + \left[\lambda_0\text{Bias}\left(\hat{\theta}^U\right) + (1-\lambda_0)\text{Bias}\left(\hat{\theta}^V\right)\right]^2,
		\end{eqnarray*} 
		where the optimal value of $\lambda$ is given by the solution to the following equation:
		\begin{eqnarray*}
			\frac{\partial}{\partial \lambda_0}\text{MSE}\left(\hat{\theta}^C_{\lambda_0}\right) &=& \lambda_0\left[\text{Var}\left(\hat{\theta}^U\right) + \text{Var}\left(\hat{\theta}^V \right) - 2\text{Cov}\left(\hat{\theta}^U, \hat{\theta}^V \right)  \right. \\ & & \, \, + \left. \left(\text{Bias}\left(\hat{\theta}^U \right) - \text{Bias}\left(\hat{\theta}^V \right)\right)^2\right] -  \text{Var}\left(\hat{\theta}^V \right)  \\ & & \, \, + \,\text{Cov}\left(\hat{\theta}^U, \hat{\theta}^V \right) +  \text{Bias}\left(\hat{\theta}^V \right) \left[\text{Bias}\left(\hat{\theta}^U\right)- \text{Bias}\left(\hat{\theta}^V \right) \right] \\ & = &  0.
		\end{eqnarray*}
		Solving the previous equation yields the expression of $\lambda_0$ in Equation (\ref{eq:lambda_opt}).
	\end{proof}
	
	Due to the unavailability of the exact values of each term in the expression of $\lambda_0$ in Equation (\ref{eq:lambda_opt}), it is proposed to estimate $\lambda_0$ using bootstrap, specifically, the smoothed bootstrap, see Algorithm 1.
	
	\begin{algorithm}[htbp]
		\small
		\caption{Bootstrap procedure for computing $\hat{\lambda}_0$}\label{pseudo:lambda}
		\begin{algorithmic}[1]
			\State \textbf{Input} a sample $\{(X_i,Y_i)\}_{i=1}^n$. 
			\State Select the bandwidths $h_1,h_2$
			\State Compute $\hat{\theta}$ = dCorV(\textbf{X,Y}).  
			\For{$b = 1$ to $B$} \{Generate B bootstrap samples\}
			\State Sample $n$ values, $U_1,\dots, U_n$, from $[1,\dots,n]$ with replacement.
			\State Generate $W_1^1,\dots,W_n^1$ iid and $W_1^2,\dots,W_n^2$ with common density $K$.
			\For{$i = 1$ to $n$} 
			\begin{eqnarray*}
				X_i^* &=& X_{[U_i]} + h_1W_i^1 \\ Y_i^* &=& Y_{[U_i]} + h_2W_i^2
			\end{eqnarray*}
			\EndFor
			\State Compute $\hat{\theta}^U_b =  $ dCorU$(\textbf{X}^*,\textbf{Y}^*)$ and $\hat{\theta}^V_b =  $ dCorV$(\textbf{X}^*,\textbf{Y}^*)$.
			\EndFor
			\State Compute Var$\left(\hat{\theta}^U\right)$, Var$\left(\hat{\theta}^V\right)$, Bias$\left(\hat{\theta}^U\right) = \frac{1}{B}\sum_{b = 1}^B \hat{\theta}^U_b - \hat{\theta}$, Bias$\left(\hat{\theta}^V\right) = \frac{1}{B}\sum_{b = 1}^B \hat{\theta}^V_b - \hat{\theta}$ and Cov$\left(\hat{\theta}^U, \hat{\theta}^V\right)$. 
			\State Obtain $\hat{\lambda}_0$ using the expression in Equation (\ref{eq:lambda_opt}).
		\end{algorithmic}
	\end{algorithm}
	
	The smoothed bootstrap relies on $K$, a kernel function (typically a symmetric density around zero), and $h > 0$, a smoothing parameter referred to as the bandwidth. In this case, $h_1,h_2$. The bandwidth regulates the size of the environment used for estimation. It is customary to stipulate that the kernel function $K$ be non-negative and have an integral of one. Additionally, it is often expected for $K$ to be symmetric. While the selection of the function $K$ does not significantly influence the properties of the estimator (aside from its regularity conditions like continuity, differentiability, etc.), the choice of the smoothing parameter is crucial for accurate estimation. In essence, the size of the environment utilized for nonparametric estimation must be appropriate-not excessively large nor too small. \\
	
	In \proglang{R} it is possible to use the \proglang{density()} function of the base package to obtain a kernel-like estimate of the density (with the bandwidth determined by the bw parameter), some of the methods implemented include Silverman's \cite{silverman1986density} and Scott's \cite{scott} rules of thumb, unbiased and biased cross-validation methods and direct plug-in methods, among others. Although it is possible to use implementations from other packages.
	
	\section{Simulation study} \label{sec:simulation}
	
	The simulation study employs the \textbf{\proglang{dcortools}} package, specifically utilizing the \proglang{distcor} function. To estimate the distance correlation through $\mathrm{dCorU}$, the code used is \proglang{distcor(X,Y,bias.corr = T)}. For computing $\mathrm{dCorV}$, there are two options:  \proglang{distcor(X,Y,bias.corr = F)} or simply \proglang{distcor(X,Y)}. The study is divided into three main parts. First, a comparison is conducted between the original estimators. Subsequently, a comparison is made between the dCorU estimator and the proposed alternatives to mitigate negative values when dealing with dependence or very low dependence and small sample sizes. Finally, a comprehensive comparison is performed, encompassing all scenarios. This includes the original estimators, comparisons of the proposed alternatives for dCorU, and estimates for the proposed convex linear combination. \\
	
	For each scenario, the mean, bias, variance and mean squared error (MSE) are computed (see Appendix \ref{secA1}). Each simulation is repeated for 1000 Monte Carlo iterations, along with three sample sizes: 100, 1000, and 10000. To compare the efficiency (MSE) and computational time of each estimator, three models with varying levels of dependence are considered, FGM, bivariate normal and a nonlinear models, defined as follows.
	
	\subsection{Models}
	
	\textbf{FGM model.} The first model corresponds to a copula. One of the most popular parametric families of copulas is the Farlie-Gumbel-Morgenstern (FGM) family, which is defined by 
	\begin{eqnarray*}
		C^{FGM}(x,y) &=& xy[1+\theta(1-x)(1-y)], \quad \quad \theta \in [-1,1]
	\end{eqnarray*}
	with copula density given by 
	\begin{eqnarray}\label{eq:density_FGM}
		c^{FGM}(x,y) &=& 1+\theta(2x-1)(2y-1), \quad \quad \theta \in [-1,1],
	\end{eqnarray}
	where $f_X(x), f_Y(y) \sim U(0,1)$. A well-known limitation of this family is that it does not allow the modeling of high dependences since Pearson correlation coefficient is limited to $\rho = \frac{\theta}{3}\in\left[-\frac{1}{3},\frac{1}{3}\right]$. Accurate calculation of the distance covariance is achieved using Equation (\ref{eq:expectations}), where the density function corresponds to Equation (\ref{eq:density_FGM}). Similarly, calculations were performed for $\dcov(X)$ and $\dcov(Y)$. As a result, $\dcor(X,Y) = \frac{|\theta|}{\sqrt{10}}$ is obtained. 	It is important to note that the dependence is not strong, specifically, $0 \leq \dcor \leq 0.31622$ (see Figure \ref{fig:Sample}).\\ 
	
	\textbf{Bivariate normal model}. The exact result for the distance correlation in this case is provided by \cite{szekely_measuring_2007} and is expressed as a function of the Pearson correlation coefficient. If $(X,Y)$ has bivariate normal distribution with unit variance each, then, the squared distance correlation is given by:
	\begin{eqnarray*}
		\dcor^2(X,Y) = \frac{\rho \arcsin{\rho}+\sqrt{1-\rho^2}-\rho \arcsin{\rho/2}-\sqrt{4-\rho^2}+1}{1+\pi/3-\sqrt{3}}.
	\end{eqnarray*}
	
	In contrast to the previous model, the bivariate normal model allows complete dependence, resulting in $\dcor = 1$ when $\rho=1$ or  $\rho=-1$. This feature facilitates examining the performance of the estimators under high levels of dependence. \\ 
	
	\textbf{Nonlinear model.} Let $(X, Y)$ be a bivariate random variable with density  $$f_{X,Y}(x,y)=c\left[1-\left(y-4\left(x-\frac{1}{2}\right)^2\right)^2\right]^kI_{[0,1]}(x)I_{[0,1]}(y),$$ with $k \in \mathbb{Z}$, controls the degree of dependence (it increases with $k$), and $c$ is a fixed value that depends on the value of $k$. For this model the $k$'s used are $0$ (independence), $2,4,8$ and $16$ (strong dependence). An initial observation is that even at a relatively low level of dependence, such as $k = 4$, the model's behavior becomes discernible (see Figure \ref{fig:Sample}). It is important to notice that, in this model, he value of the distance correlation when $X$ and $Y$ are totally dependent (i.e. when $k \to \infty$) is not 1, but $\dcor \approx 0.41$. So, unlike the linear models (FGM and normal bivariate), in this model low values of $\dcor$ ($0.22 \leq \dcor \leq 0.41$) could indicate a strong relationship between $X$ and $Y$. Figure \ref{fig:Sample} displays four samples drawn along with their respective distance correlation. The lines represent the conditional mean $E[Y|X = x]$. \\
	\begin{figure}[htbp]
		\centering
		\begin{subfigure}{1\textwidth}
			\includegraphics[width=\textwidth]{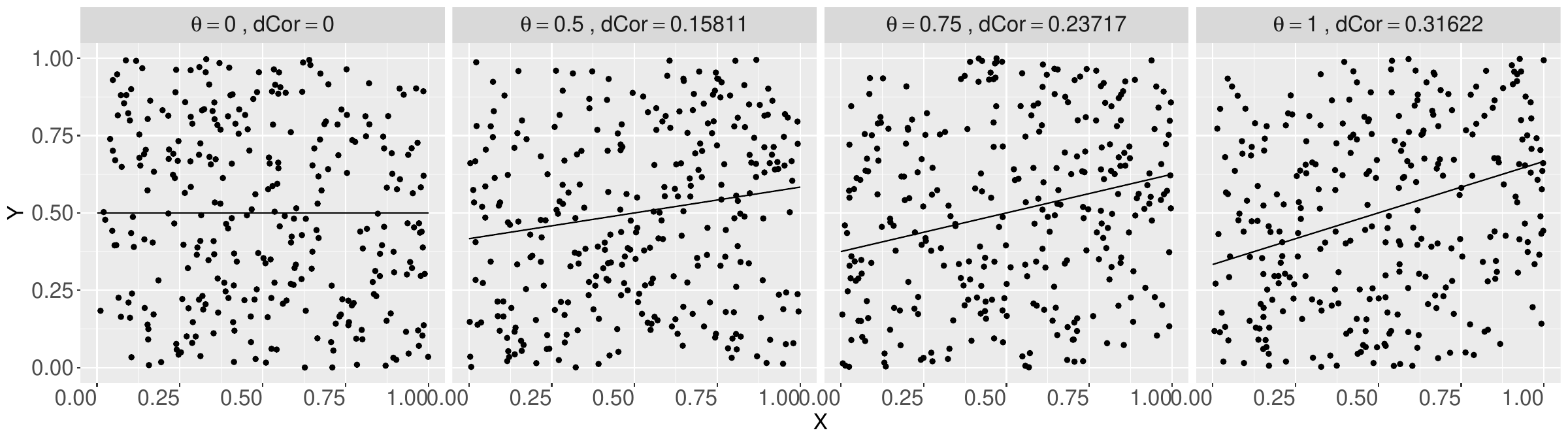}
		\end{subfigure}
		\begin{subfigure}{1\textwidth}
			\includegraphics[width=\textwidth]{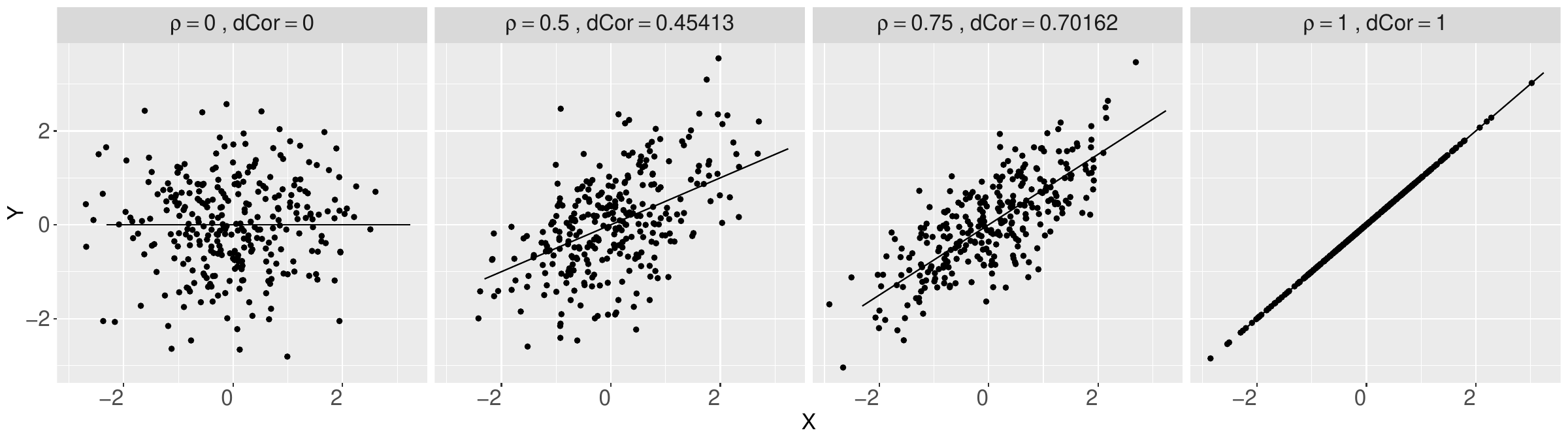}
		\end{subfigure}
		\begin{subfigure}{1\textwidth}
			\includegraphics[width=\textwidth]{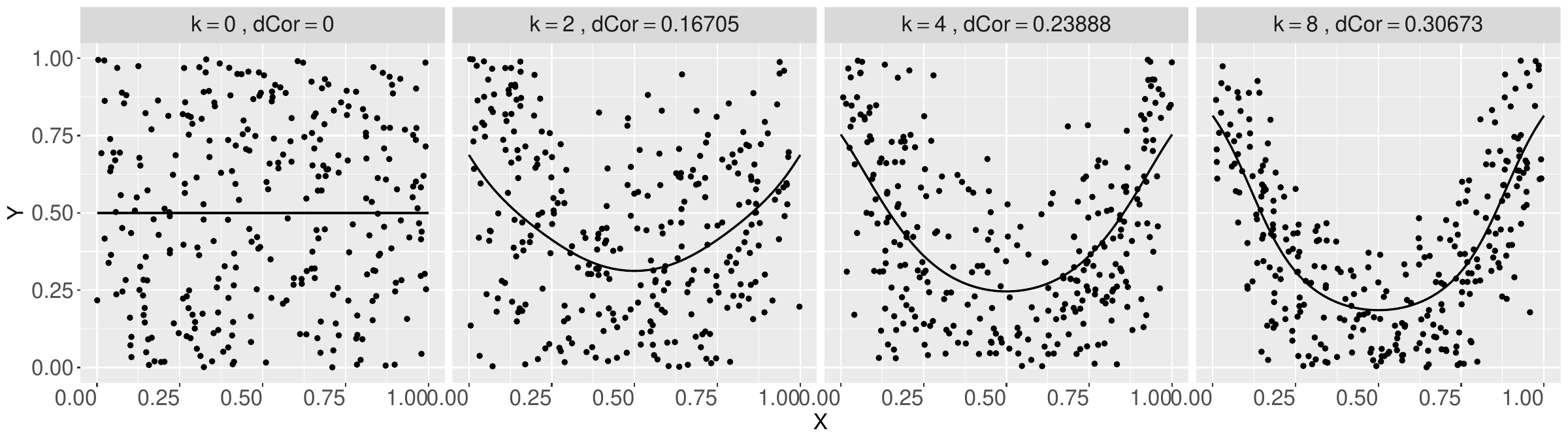}
		\end{subfigure}
		\caption{Samples ($n = 300$) generated for each model: FGM in the first row, bivariate normal in the second row, and nonlinear in the last row, across different values of $\theta, \rho,$ and $k$, respectively, along with their corresponding distance correlation.}
		\label{fig:Sample}
	\end{figure}
	
	\subsection{Comparison between dCorU and dCorV estimators} \label{sec:negatives}
	
	In this part of the simulation study, the original estimators, dCorU and dCorV, are compared. Firstly, the comparison for the MSE of dCorU and dCorV across the different sample sizes for each distance correlation for the three  models is shown in Figure \ref{fig:MSE}. Each simulation is repeated for 1000 Monte Carlo iterations, along with three sample sizes: 100, 1000, and 10000. Since the \textbf{\proglang{dcortools}} package is being used, then dCorU turns out to be dCorU = sign(dCorU$^2$) $\sqrt{|\text{dCorU}^2|}$. \\
	
	\begin{figure}[htbp!]
		\centering
		\begin{subfigure}{1\textwidth}
			\includegraphics[width=\textwidth]{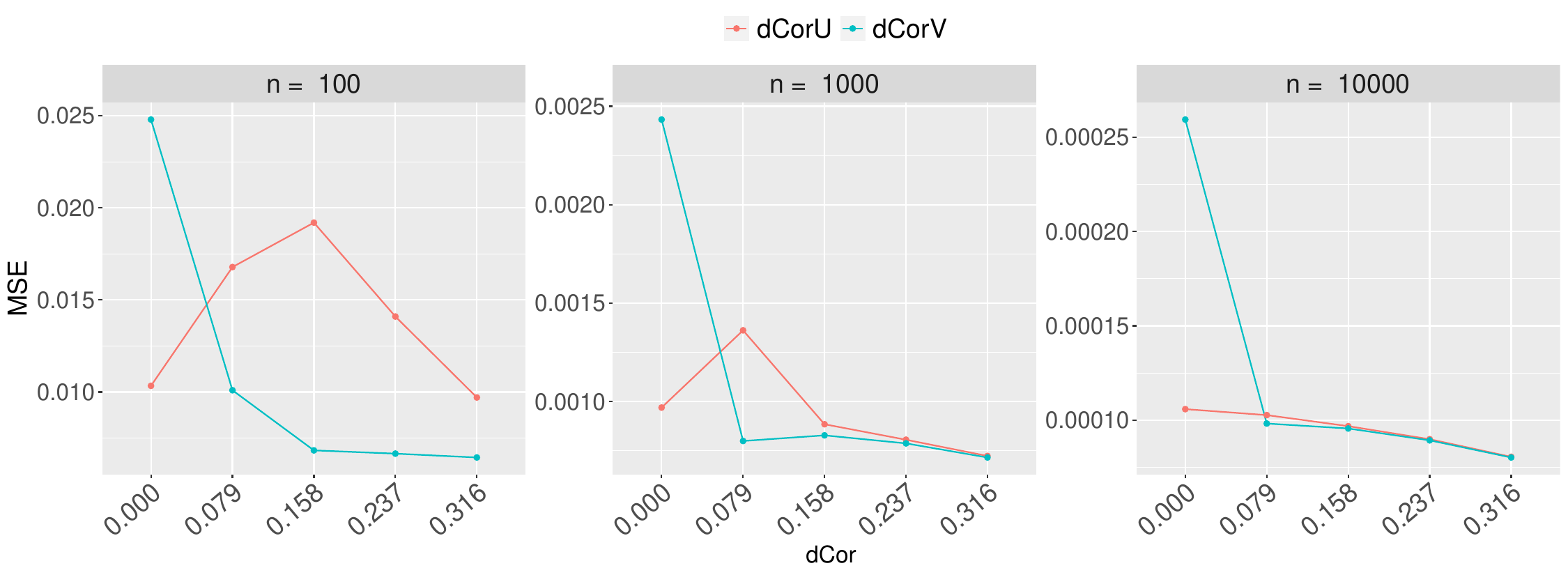}
			\caption{FGM-Model for distance correlation values from $\dcor = 0 \,\, (\theta = 0)$ to $\dcor = 0.316 \,\, (\theta = 1)$.}
			\label{fig:FGM_MSE}
		\end{subfigure}
		\begin{subfigure}{1\textwidth}
			\includegraphics[width=\textwidth]{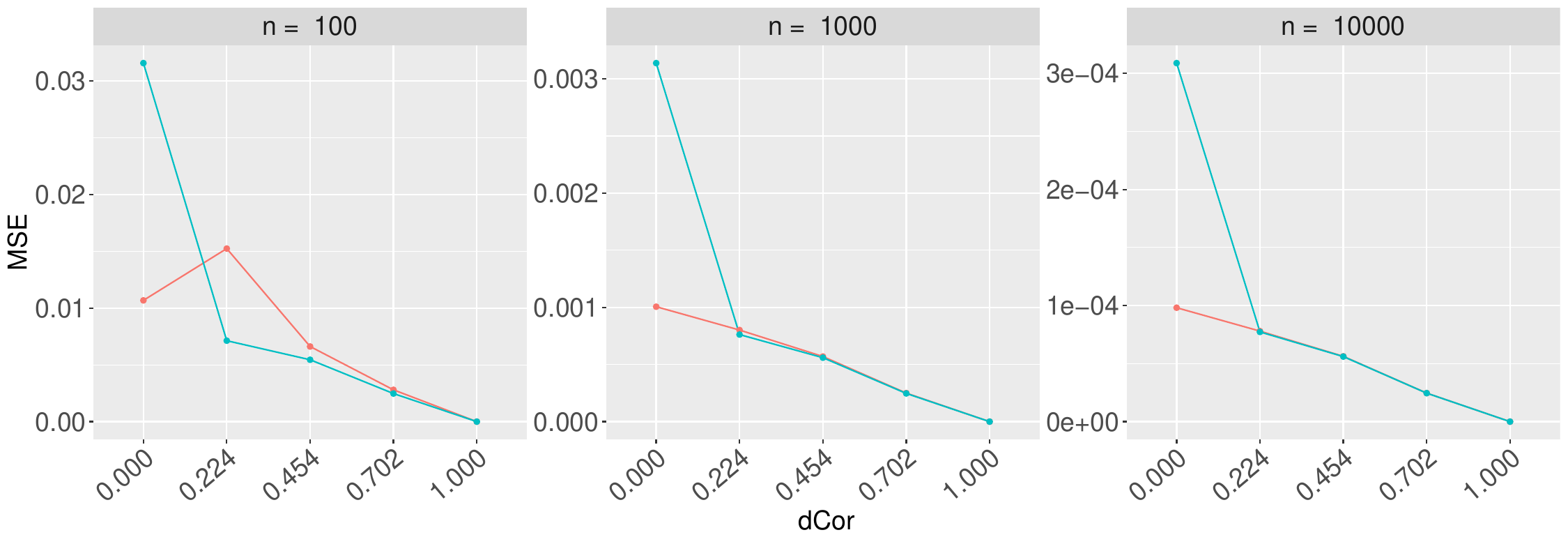}
			\caption{Bivariate normal model for distance correlation values from $\dcor = 0 \, (\rho = 0)$ to $\dcor = 1 \, (\rho = 1)$.}
			\label{fig:BN_MSE}
		\end{subfigure}
		\begin{subfigure}{1\textwidth}
			\includegraphics[width=\textwidth]{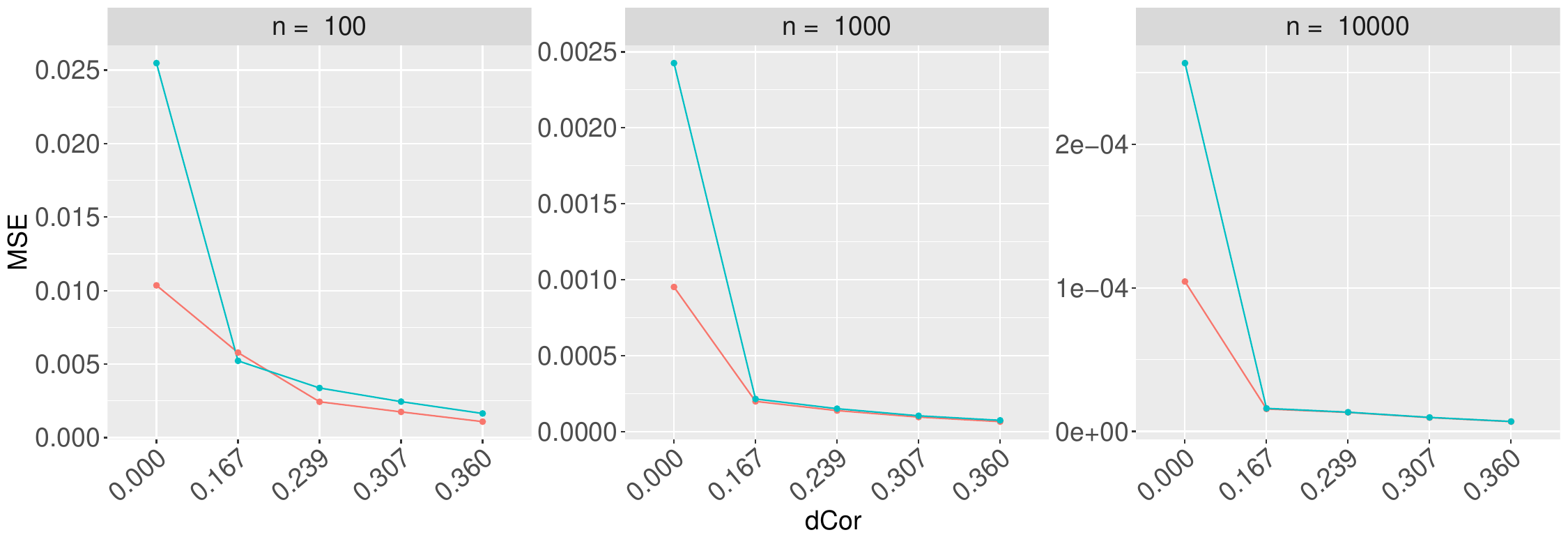}
			\caption{Nonlinear model for distance correlation values from $\dcor = 0 \,\, (k = 0)$ to $\dcor = 0.36\,\, (k = 16)$.}
			\label{fig:NLM_MSE}
		\end{subfigure}
		\caption{MSE of dCorU and dCorV under the different models with sample sizes $n = 100$ (left), $n = 1000$ (center) and $n = 10000$ (right) for different distance correlation values.}
		\label{fig:MSE}
	\end{figure}
	
	For the FGM model the differences between the two estimators become more significant across all values of $\theta$. Under independence ($\dcor = 0$), dCorU outperforms dCorV in terms of MSE for all three sample sizes. However, when there is dependence, dCorV shows better results, even in the presence of a small degree of dependence. As the level of dependence increases, the MSE of both estimators decreases and the differences fade away.\\ 
	
	The MSE obtained for the bivariate normal model is shown in Figure \ref{fig:BN_MSE}. The conclusions are similar to the previous case. Under independence ($\dcor = 0$), dCorU emerges as the superior estimator. However, as soon as there is no independence, the conclusion is exchanged, and dCorV is the best option. As dependence gets stronger ($\dcor \geq 0.454$), the sample size becomes less influential, and both estimators tend to converge and provide the same value. Furthermore, for larger sample sizes, a moderate dependence ($\dcor > 0.224$) is sufficient to observe similar behavior between the estimators. It is worth noting that under complete dependence ($\dcor = 1$), the estimates were the same for each sample size. \\ 
	
	The results for the nonlinear model are shown in Figure \ref{fig:NLM_MSE}. Similar to the previous models, under independence  the optimal estimate is provided by dCorU. However, dCorV estimator does not systematically emerge as the superior choice at any specific level of dependence. Both estimators exhibiting similar behavior with large sample sizes and at dependence holds true for this model as well. \\ 
	
	The results of the bias and variance of both estimators are presented in Appendix \ref{secA1}. They provide information that is consistent with the conclusions drawn from the MSE analysis. In particular, the most substantial difference between the two estimators arises in the case of a small sample size ($n = 100$) and weak dependence or independence. It is observed that dCorU shows a negative bias that approaches zero with increasing levels of dependence. Similarly, dCorV also shows a decreasing bias, although always positive. On the other hand, the variance of dCorV is comparatively smaller than the variance of dCorU. In contrast, for larger sample sizes ($n = 1000, 10000$), the results show insignificant differences between the two estimators.\\

	The computational times for each estimator are in Table \ref{tab:time}. The characteristics of the computer equipment used are the following ones: CPU 12th Gen Intel(R) Core(TM) i7-1280P 2.00 GHz and RAM 16 GB. Note that there are no significant differences in the computation times between estimators. It is important to highlight that these times were obtained with the \textbf{\proglang{dcortools}} package. Where both the U-statistic and the V-statistic of distance correlation can be calculated in $\mathcal{O}(n \log n)$ steps. In fact, the algorithm in \cite{huo_fast_2016} can be straightforward extended to the V-statistic version. If another package were used, for example the \textbf{\proglang{energy}} package, the computation times of both methods will be different.
	
	\begin{table}[htbp!]
		\begin{tabular}{ c   c  c |  c c |c c }
			\multicolumn{1}{c}{} & \multicolumn{2}{c}{$n =$ 100} & \multicolumn{2}{ c }{$n =$ 1000} & \multicolumn{2}{ c }{$n =$ 10000}  \\
			& dCorU & dCorV & dCorU & dCorV & dCorU & dCorV \\ \cmidrule{2-7} 
			Time &  0.36 & 0.33 & 0.63 & 0.58 & 2.22 & 2.51 \\ 
		\end{tabular} 
		\caption{Computational time in secs for 1000 samples.}
		\label{tab:time}
	\end{table} 
	
	\subsection{Alternatives to dCorU}
	As mentioned above, when working under independence or with very low dependencies, then the dCorU estimator frequently cannot be computed because of negative estimates of the squared covariance. Table \ref{tab:negatives} provides the percentage of negative values obtained for different sample sizes across the three models under varying levels of dependence. This section compares the dCorU estimator to the alternatives proposed in Section 2, dCorU(A) and dCorU(T). Figure \ref{fig:truncated} illustrates the comparison of the MSE between dCorU and the proposed methods for each model. \\
	
	\begin{table}[htbp!] 
		\begin{adjustbox}{width=\textwidth}
			\begin{tabular}{c c c c c  c c c c c c c c c c}
				\multicolumn{5}{c}{\textbf{FGM-Model}} & \multicolumn{5}{c}{\textbf{Bivariate Normal}} & \multicolumn{5}{c}{\textbf{Nonlinear model}} \\
				\toprule
				\multicolumn{2}{c}{ } & \multicolumn{3}{c}{\% of negative values} & \multicolumn{2}{c}{ } & \multicolumn{3}{c}{\% of negative values} & \multicolumn{2}{c}{ } & \multicolumn{3}{c}{\% of negative values} \\
				\multicolumn{2}{c}{ } & \multicolumn{3}{c}{$n$} & \multicolumn{2}{c}{ } & \multicolumn{3}{c}{$n$} & \multicolumn{2}{c}{ } & \multicolumn{3}{c}{$n$}\\
				\cmidrule{3-5} \cmidrule{8-10}
				\multicolumn{1}{c}{$\theta$} & $\dcor$ & $100$ & $1000$ &  $10000$ & \multicolumn{1}{c}{$\rho$} & $\dcor$ & $100$ & $1000$  & $10000$ & \multicolumn{1}{c}{$k$} & $\dcor$ & $100$ & $1000$ &  $10000$\\ 
				\hline
				\hline
				0    & 0     & 65.4 &  65.6 & 62.3 & 0    & 0     & 59.2 & 62.6 & 64.3 & 0  & 0     & 60.6 & 65.4 & 62.7\\
				0.25 & 0.079 & 50.3 &  5.1  & 0    & 0.25 & 0.224 & 9.3  & 0    & 0    & 2  & 0.167 & 4.7  & 0    & 0   \\
				0.5  & 0.158 & 26   &  0    & 0    & 0.5  & 0.454 & 0    & 0    & 0    & 4  & 0.239 & 0.1  & 0    & 0   \\ 
				0.75 & 0.237 & 6.5  &  0    & 0    & 0.75 & 0.702 & 0    & 0    & 0    & 8  & 0.307 & 0    & 0    & 0   \\ 
				1    & 0.316 & 1    &  0    & 0    & 1    & 1     & 0    & 0    & 0    & 16 & 0.360 & 0    & 0    & 0\\ 
				\bottomrule
			\end{tabular}
		\end{adjustbox}		  
		\caption{Percentage of negative values for $\mathcal{U}_n^2(X,Y)$ obtained when computing dCorU$^2$ in (\ref{eq:dcorU}) for each of the models across different scenarios. }
		\label{tab:negatives}
	\end{table}

	The FGM-Model presents negative values across all the studied values of $\theta$ when $n = 100$. As the sample size increases, the problem of negative values only appears when the dependence is very weak ($n = 100$) or just under independence ($n = 10000$), see Table \ref{tab:negatives}. Despite the notably high percentage of negative values in the case of independence, the differences between the estimators are not substantial.  However, the most accurate estimation is achieved using dCorU(T), given that $\dcor = 0$. Conversely, under dependence, dCorU(A) provides a more accurate estimation. Nevertheless, as the level of dependence strengthens, the MSE values for all three estimations tend to the same value. \\
	
	In the case of the bivariate normal model (Figure \ref{fig:truncated}), negative values are encountered only under independence ($\dcor = 0$) or low levels of dependence ($\dcor = 0.22$), both with $n=100$. Under strong dependence or large sample sizes, the computation of dCorU does not exhibit the issue of negative values (at least with $n \geq 100$). The conclusions are parallel to those of the previous model. In the case of independence, dCorU(T) emerges as the optimal estimator, while in another scenario, it is dCorU(A). In the remaining cases, all three estimators yield comparable MSE. Furthermore, for $n=10000$, the issue of negative values arises in the independence scenario (Table \ref{tab:negatives}), but the differences remain insignificant.\\
	
	Finally, in the case of the nonlinear model, negative values are observed for weak dependence scenarios ($k = 0,2,4$) with $n = 100$. However, as the sample size increases ($n = 1000, 10000$), the problem only arises under independence ($k = 0$). Note that the conclusions are similar to the two previous cases.  Under independence, the optimal estimator is dCorU(T), whereas in one of the other cases, it is dCorU(A). For the remaining cases, all three estimators yield the same MSE (Figure \ref{fig:truncated}). 
	
	\begin{figure}[htbp!]
		\centering
		\begin{subfigure}{1\textwidth}
			\includegraphics[width=\textwidth]{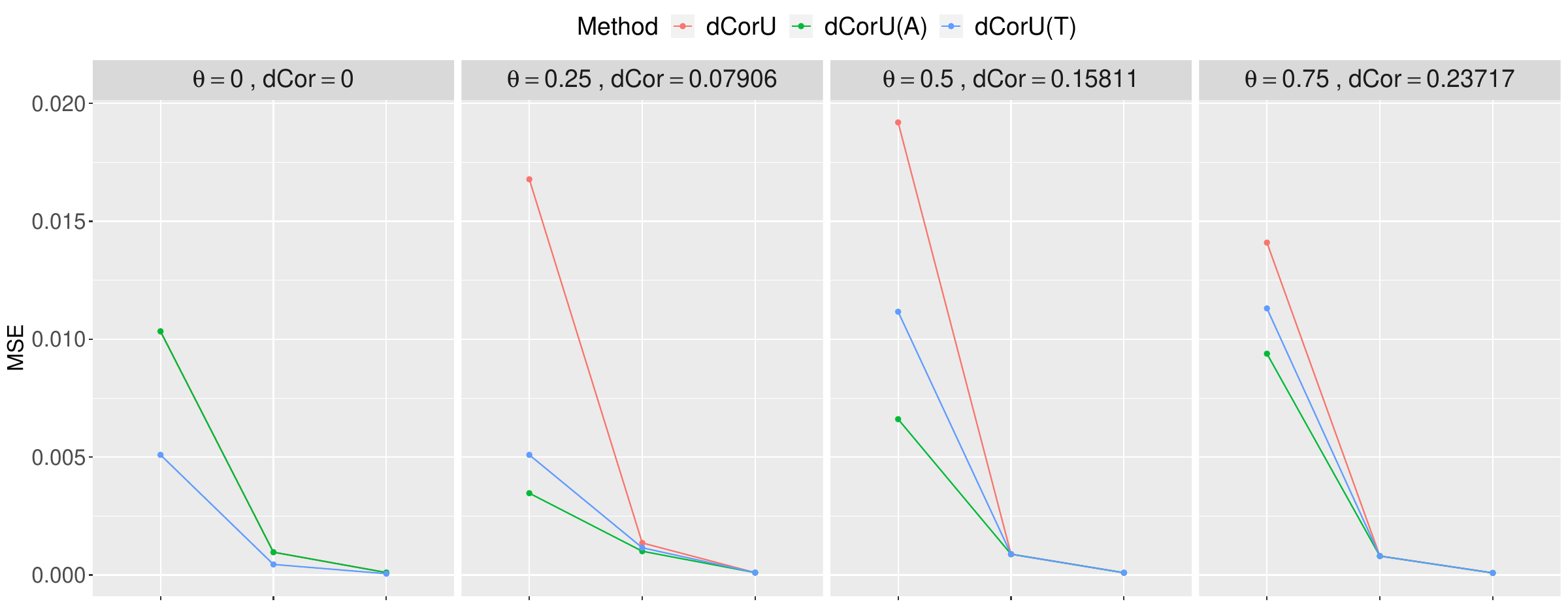}
		\end{subfigure}
		\begin{subfigure}{1\textwidth}
			\includegraphics[width=1\textwidth]{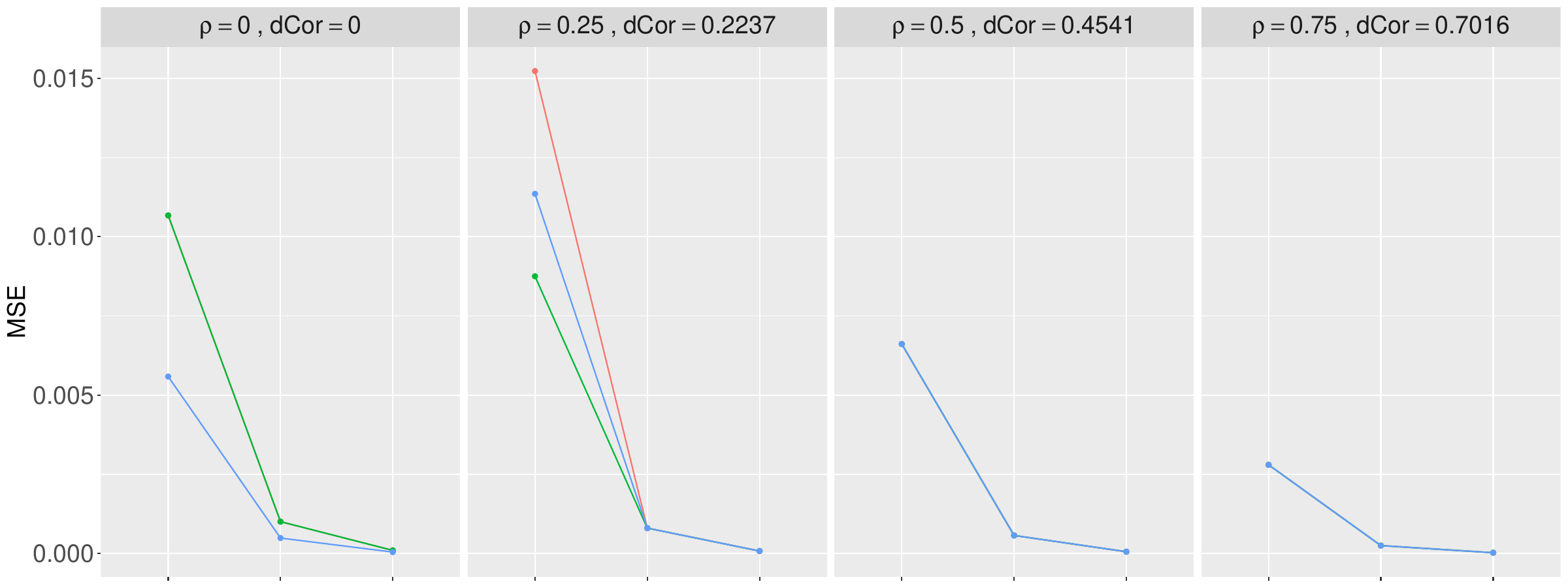}
		\end{subfigure}
		\begin{subfigure}{1\textwidth}
			\includegraphics[width=\textwidth]{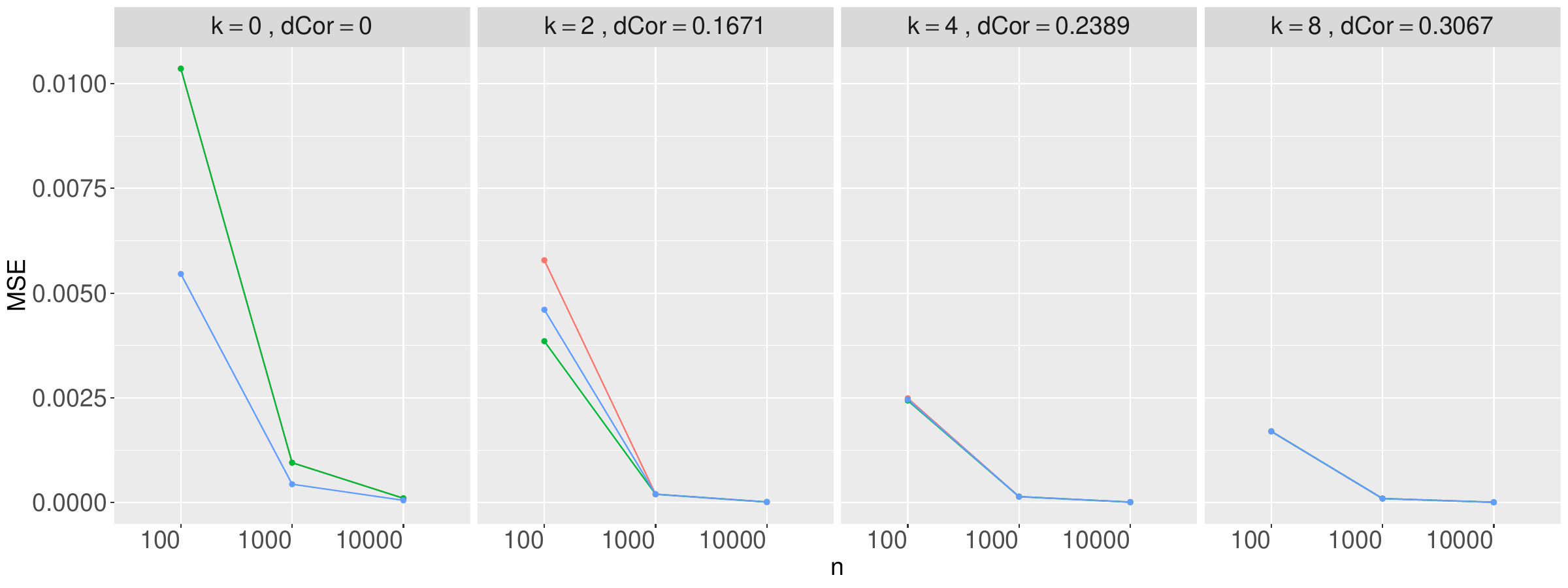}
		\end{subfigure}
		\caption{MSE values for dCorU, dCorU with absolute value (dCorU(A)), and truncated value (dCorU(T)) across three sample sizes ($n$) and three different models: the FGM model in the first row, the bivariate normal model in the second row, and the nonlinear model in the last row. The corresponding distance correlation values are also presented.}
		\label{fig:truncated}
	\end{figure}
	
	\subsection{New estimator for the distance correlation} \label{sec:studyC}
	
	This section explores and compares the original estimators (dCorU and dCorV) with proposed variations, including dCorU(A) and dCorU(T), along with their corresponding convex combinations computed with the optimal value of the parameter, $\lambda_0$ (Eq. \ref{eq:lambda_opt}), as follows: 
	\begin{eqnarray*}
		\text{dCor}_{\lambda_0} &=& \lambda_0\text{dCorU} + (1 - \lambda_0) \text{dCorV}, \\ \text{dCor(A)}_{\lambda_0} &=& \lambda_0\text{dCorU(A)} + (1 - \lambda_0) \text{dCorV}, \\ \text{dCor(T)}_{\lambda_0} &=& \lambda_0\text{dCorU(T)} + (1 - \lambda_0) \text{dCorV}.
	\end{eqnarray*}
	
	The comparisons are conducted across two scenarios for the convex linear combination. The first scenario involves estimating the optimal $\lambda_0$ (Eq. (\ref{eq:lambda_opt})) using Algorithm \ref{pseudo:lambda}, while the second scenario utilizes the optimal value of $\lambda_0$ in Lemma \ref{le:lambda}, where the bias, variance covariance terms in Equation (\ref{eq:lambda_opt}) are approximated using 1000 Monte Carlo repetitions. The study focuses on the behavior of the MSE in each model under different levels of dependence and sample size $n = 100$ (see Figure \ref{fig:lambda_opt}). \\
	
	The value of $\hat{\lambda}_0$ was obtained using Algorithm \ref{pseudo:lambda} with 1000 bootstrap iterations. The bandwidths, considered $h_1 = h_2$ for simplicity, were chosen as the value of the grid between 0.0025 and 0.32 with the lowest MSE. A simulation study showing the effect of the bandwidths is included in Appendix \ref{secA1}. It is evident that the bandwidth has an effect only under independence conditions, with lower bandwidth values corresponding to lower MSE. When the level of dependence is moderate, the impact of the bandwidth becomes less significant, and under high dependence, the results for different bandwidth values tend to converge.\\
	
	\begin{figure}[htbp]
		\centering
		\includegraphics[width=1\textwidth]{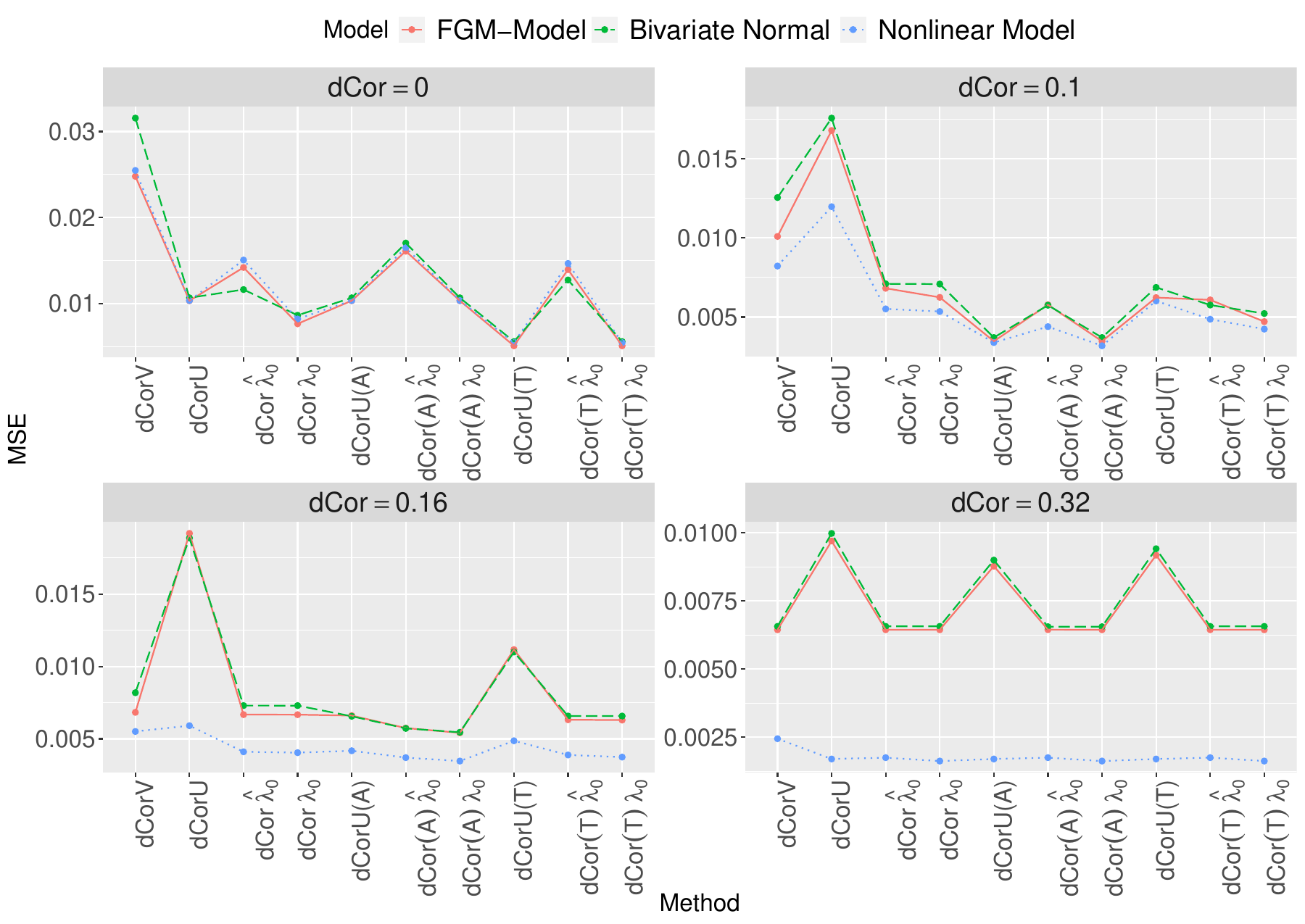}
		\caption{Comparison of the MSE between the different estimators and their respective convex linear combination and sample size $n = 100$. Specifically, $\text{dCor}_{\hat{\lambda}_0}$ denotes the combination $\hat{\lambda}_0$dCorU + $\left(1 - \hat{\lambda}_0\right)$dCorV with $\hat{\lambda}_0$ estimated using 1000 bootstrap replications. Similarly, $\text{dCor}_{\lambda_0}$ denotes the combination $\lambda_0$dCorU + $(1 - \lambda_0)$dCorV using the real optimal value of $\lambda_0$ (Eq. (\ref{eq:lambda_opt})). The same naming convention applies to dCorU(A) and dCorU(T), each representing their respective combinations. The values of $\dcor$ shown are rounded values of the corresponding ones from each model.}
		\label{fig:lambda_opt}
	\end{figure}
	
	In Figure \ref{fig:lambda_opt}, the maximum value of the distance correlation considered is $\dcor = 0.31$, which corresponds to the maximum level achievable in the FGM-Model. It is notable that, even though $\dcor$ is approximately the same for all three models, the degree of dependence associated to each value of $\dcor$ is different for each scenario, especially for the nonlinear model. This causes the behavior of the estimators is not comparable over the models for a fixed value of $\dcor$. Among the main observations that can be made: under independence ($\dcor = 0$), the proposed estimator $\text{dCor}_{\lambda 0}$ outperforms dCorU and dCorV when using the optimal $\lambda_0$. The performance of the proposed estimator worsens slightly when $\lambda_0$ is estimated, but it still gives much better results than dCorV. As soon as there is no independence ($\dcor > 0$), the proposed estimator with estimated $\lambda_0$ gives comparable results to the one with the optimal $\lambda_0$, and it improves both dCorU and dCorV \\
	
	These conclusions can be observed for each of the models in Figure \ref{fig:h}, where it is noticeable that the most significant differences with each of the estimators occur under conditions of independence. Conversely, as dependence increases, the estimates tend to become more comparable. Specifically, the nonlinear model exhibits a decrease in MSE as $\dcor$ increases, while the linear models do not show such a significant decrease. \\
	
	Figure \ref{fig:Dep_results} displays the MSE for each of the models under varying levels of dependence, from independence ($\dcor = 0$) up to the maximum dependence allowed by each model. It is worth noting that the nonlinear model reaches a maximum $\dcor$ value of around 0.4134, indicating complete dependence between $X$ and $Y$. Note that the behavior observed in the nonlinear model is also evident in the normal bivariate model, wherein the MSE remains quite similar across all cases starting from $\dcor = 0.75$. Furthermore, the reduction in MSE becomes evident for each of the cases as $\dcor$ increases, eventually reaching 0 in the case of complete dependence ($\dcor = 1$). However, this is not the case for the nonlinear model, where the MSE does not reach zero for the full dependence scenario ($\dcor = 0.41$). \\
	
	It is important to mention that in cases of independence or very low dependence, dCorU(T) and its combination with $\lambda_0$ exhibit the lowest MSE. Additionally, it is worth noting that the bootstrap approximation of $\lambda_0$ shows larger discrepancies from the actual value in scenarios of independence or low dependence (see Appendix \ref{secA1}). However, when dependence is low (approximately $\dcor = 0.1$), dCorU(A) and its convex combinations $\left(\text{dCor(A)}_{\lambda_0} \text{and } \text{dCor(A)}_{\hat{\lambda}_0}\right)$ result in a lower MSE. Under moderate or high dependence ($\dcor \geq 0.24$), convex linear combinations tend to converge to the same MSE, consistently approaching the best estimate provided by the individual estimators (dCorU, dCorV, dCorU(A), dCorU(T)).\\
	
	It becomes evident that the convex linear combination using the value $\hat{\lambda}_0$ obtained through 1000 bootstrap repetitions consistently approximates the best estimates provided by the original estimators (dCorU, dCorV) across all scenarios examined. While it is true that under specific conditions, some alternative combinations appear to yield slightly smaller errors, it is important to note that distinguishing between working under independence or with very low dependence may not always be feasible.
	
	\begin{figure}[htbp!]
		\centering
		\begin{subfigure}{0.49\textwidth}
			\includegraphics[width=\textwidth]{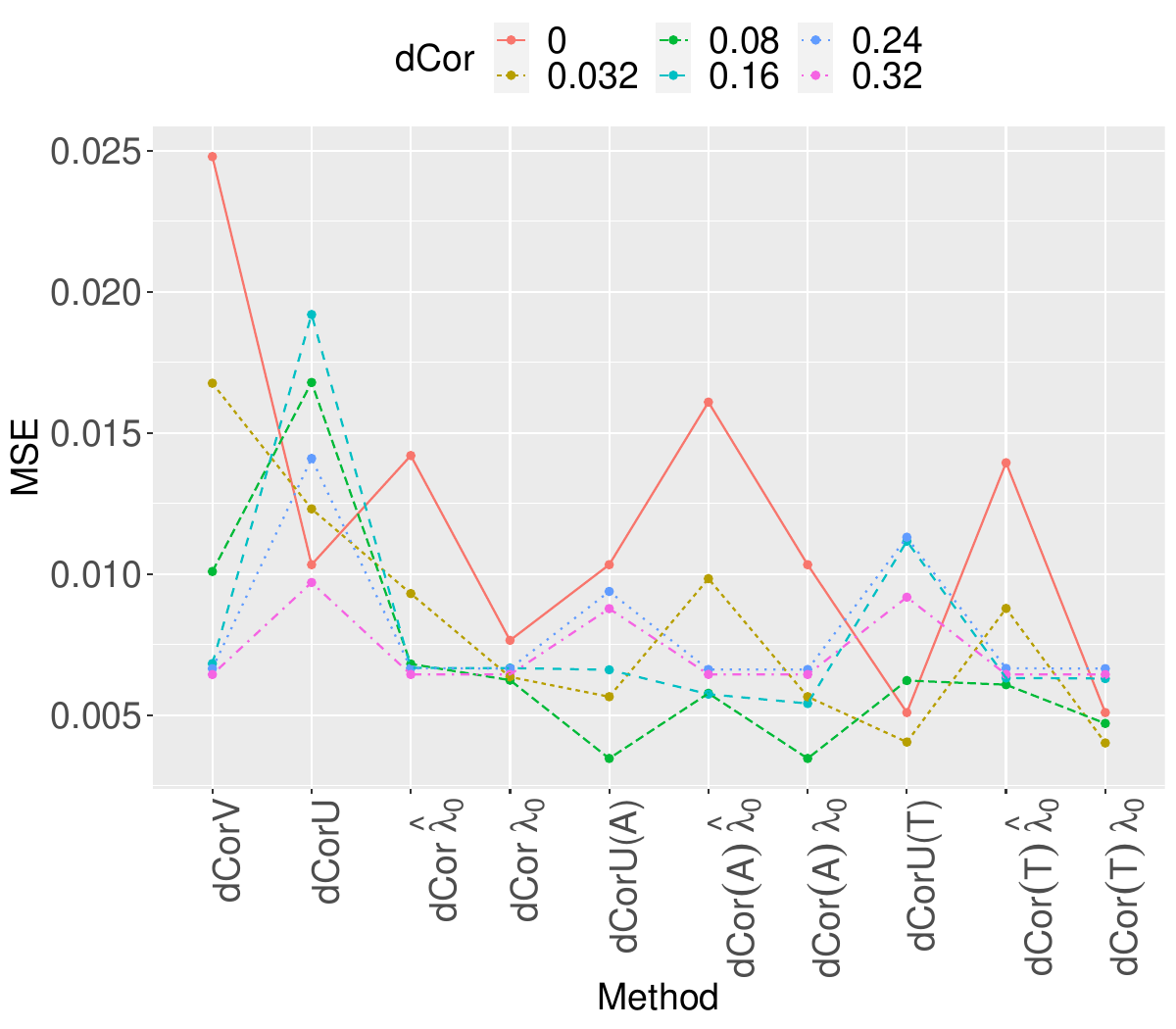}
			\caption{FGM-Model}
			\label{fig:FGM_dep}
		\end{subfigure} 
		\begin{subfigure}{0.49\textwidth}
			\includegraphics[width=\textwidth]{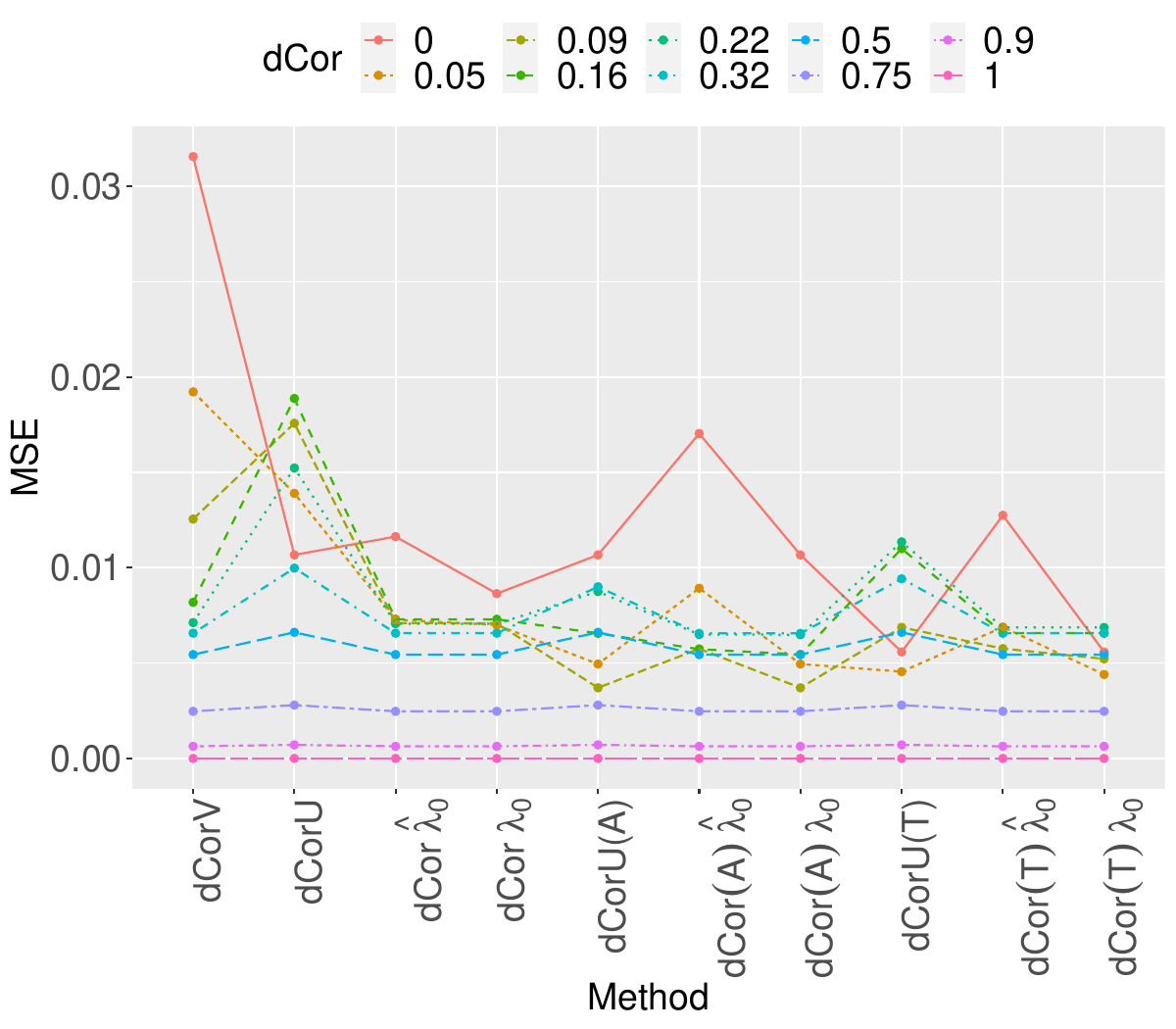}
			\caption{Bivariate normal model.}
			\label{fig:BN_dep}
		\end{subfigure} 
		\begin{subfigure}{0.5\textwidth}
			\includegraphics[width=\textwidth]{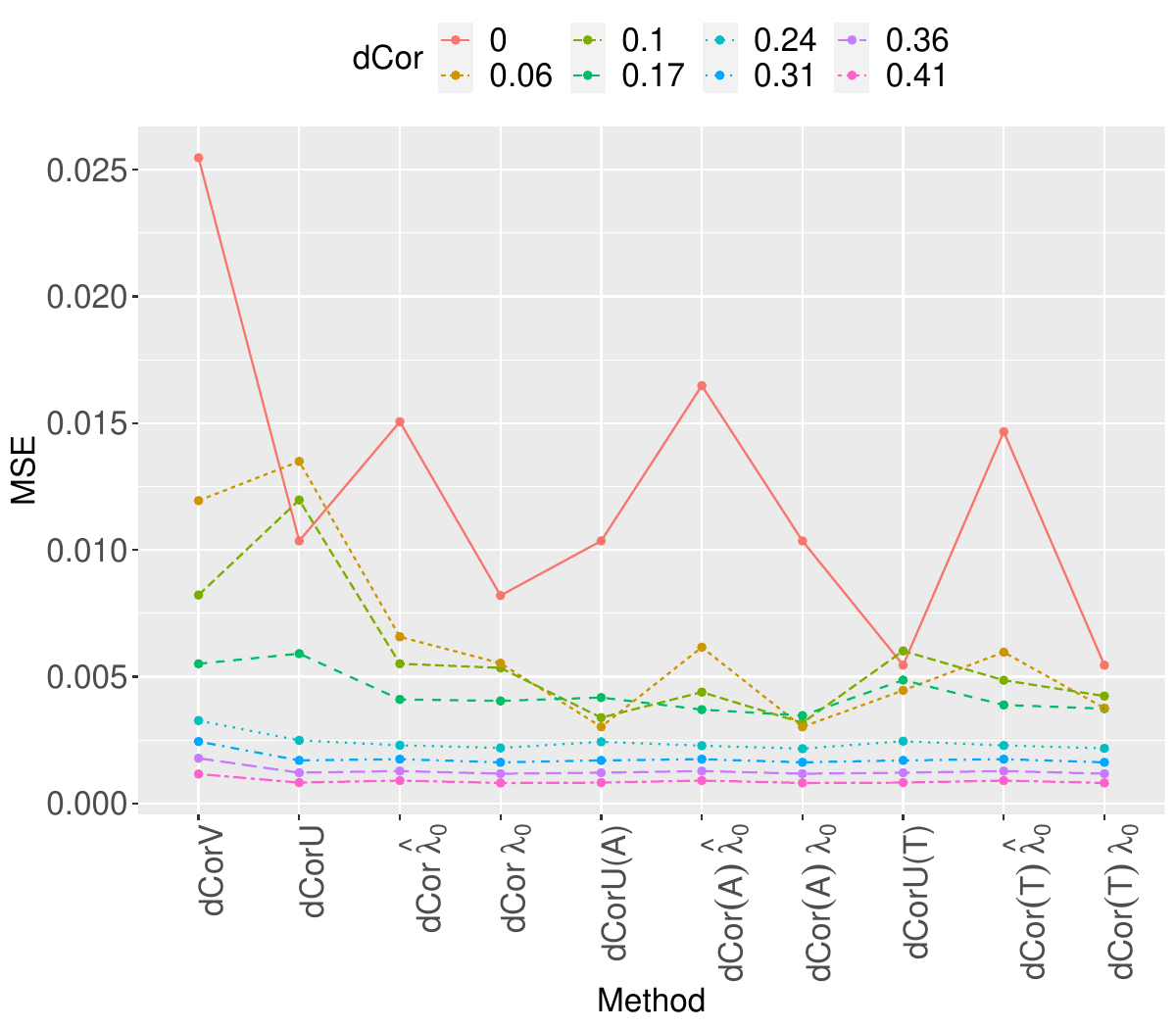}
			\caption{Non-linear model.}
			\label{fig:NLM_dep}
		\end{subfigure}
		\caption{Comparison of the MSE among the different estimators and their respective convex linear combinations performed for the provided models across various levels of dependence, from independence ($\dcor = 0$) to the strongest dependence supported by each model for $n = 100$.}
		\label{fig:Dep_results}
	\end{figure} 

	\section{Conclusions}\label{sec:Conclusions}
	
	This research addressed the problem of choosing the best method to estimate distance correlation between two vectors $X$ and $Y$.  Under independence, dCorU looks preferable over dCorV for all the models and sample sizes analyzed. However,  $\mathcal{U}_n^2(\boldsymbol{\X,\Y})$ might be negative in scenarios of independence or low dependence, particularly with small sample sizes. Consequently, this leads to an inability to calculate dCorU. This observation has led to the development of two alternative proposals, based on truncating or computing the absolute value of $\mathcal{U}_n^2(\boldsymbol{\X,\Y})$. Under independence, the superior estimation is in general when truncating.\\
	
	On the other hand, under dependence, the conclusions differ. The dCorV estimator aligns with the best results in terms of MSE for the linear models (FGM and bivariate normal). While, in the case of the nonlinear model, the dCorU estimator provides better estimates, however the differences are not relevant. Moreover, when considering the provided approaches, dCorU(A) and dCorU(T), the optimal estimator appeared to be dCorU(A), especially in cases of weak dependence. Now, in terms of computational time, both estimators, dCorU and dCorV, are similar. \\
	
	In practice, it is difficult to choose the best estimator for the distance correlation, as the choice depends on whether the relationship between $X$ and $Y$ is linear or no, and whether there is independence or not, questions than remain unknown in real-life studies. To address these complexities, a new estimator has been proposed involving the convex linear combination of the two estimators, dCorU and dCorV, as well as their respective extensions, dCorU(A) and dCorU(T). In the majority of the cases, the proposed estimator with the parameter $\lambda_0$ estimated using bootstrap, $\text{dCor}_{\hat{\lambda}_0}$, give better results than compared to using only dCorU or dCorV.  However, it is necessary to take into account that the computation time will increase as the number of bootstrap iterations increases. \\
	
	\bibliographystyle{refs-style}
	\bibliography{ref}

	\section*{Appendix Simulation study}\label{secA1}
		This appendix showcases the outcomes investigated and discussed in Section \ref{sec:simulation}. These results are illustrated through various plots depicting bias, variance, and Mean Squared Error (MSE) across the different comparisons explored. Additionally, numerical summaries for the mean, variance, bias, and MSE are provided in tabular format for each model across various scenarios.
		\subsection*{Bias and variance}
		
		This section presents the bias and variance results for each model across five levels of dependence and three sample sizes ($n = 100, 1000, 10000$). These results have been discussed in detail in Section 4 of the paper. Figure \ref{fig:FGM_var} illustrates the results obtained under the FGM model, while Figure \ref{fig:BN_var} depicts the outcomes under the bivariate normal model. Similarly, Figure \ref{fig:NLM_var} showcases the results for the nonlinear model. 
		\begin{figure}[htbp!]
			\centering
			\includegraphics[width=0.7\textwidth]{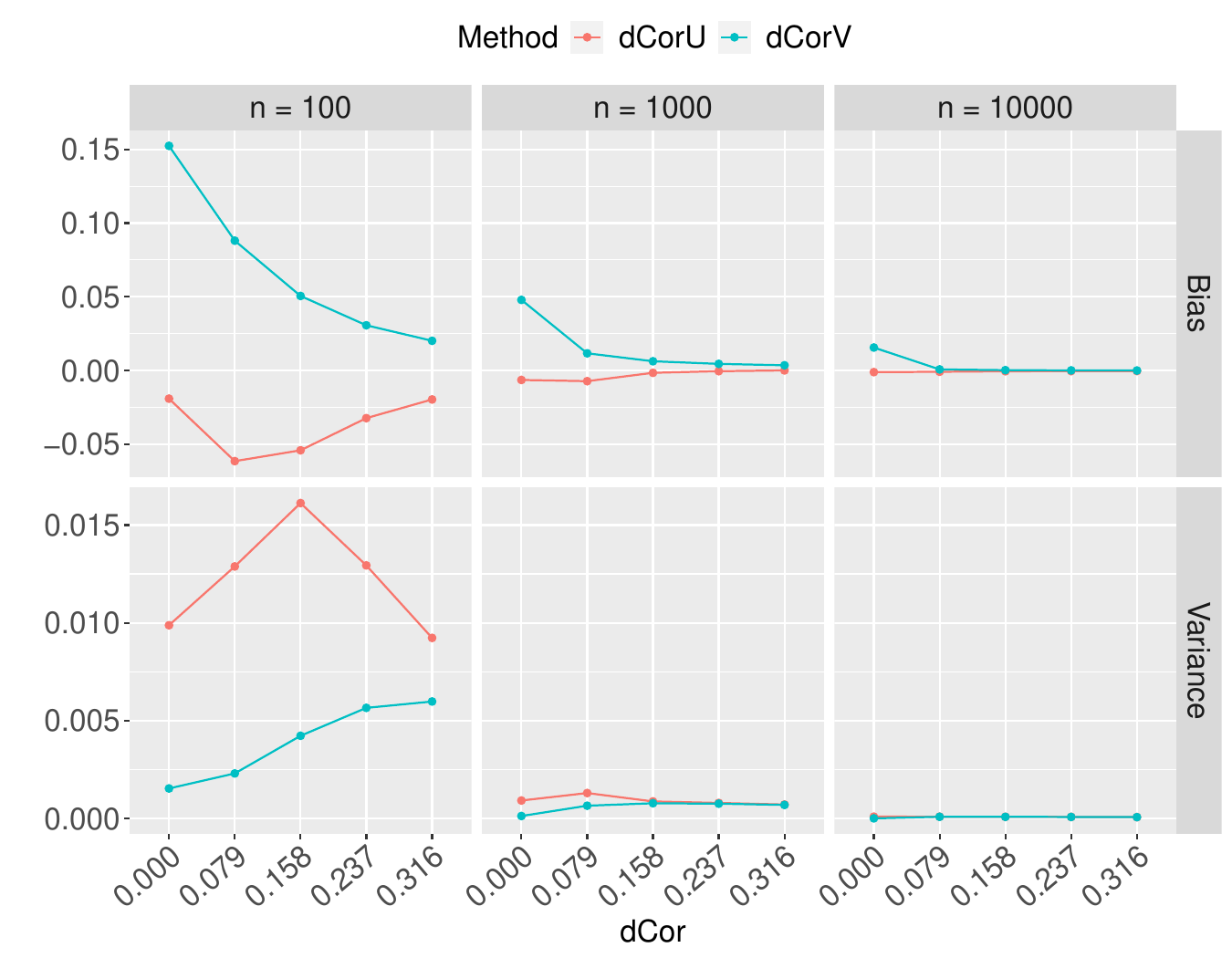}
			\caption{Bias and variance of dCorU and dCorV under the FGM-Model with different sample sizes ($n$) and distance correlation, from $\dcor = 0$ ($\theta = 0$) to $\dcor = 0.32$ ($\theta = 1$).}
			\label{fig:FGM_var}
		\end{figure}
		
		\begin{figure}[htbp!]
			\centering
			\includegraphics[width=0.7\textwidth]{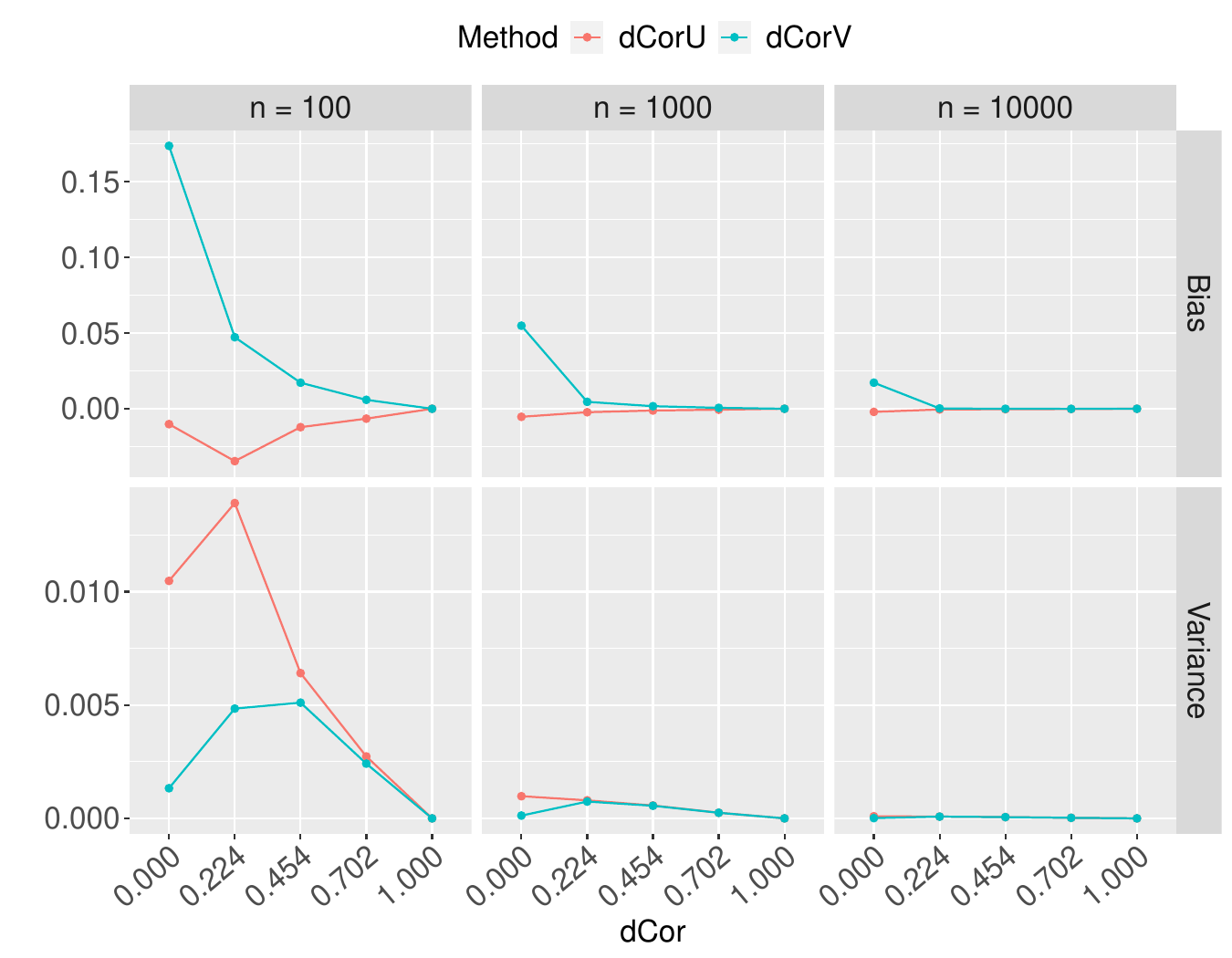}
			\caption{Bias and variance of dCorU and dCorV under bivariate normal model with different sample sizes ($n$) and distance correlation, from $\dcor = 0$ ($\rho = 0$) to $\dcor = 1$ ($\rho = 1$).}
			\label{fig:BN_var}
		\end{figure}
		
		\begin{figure}[htbp!]
			\centering
			\includegraphics[width=0.7\textwidth]{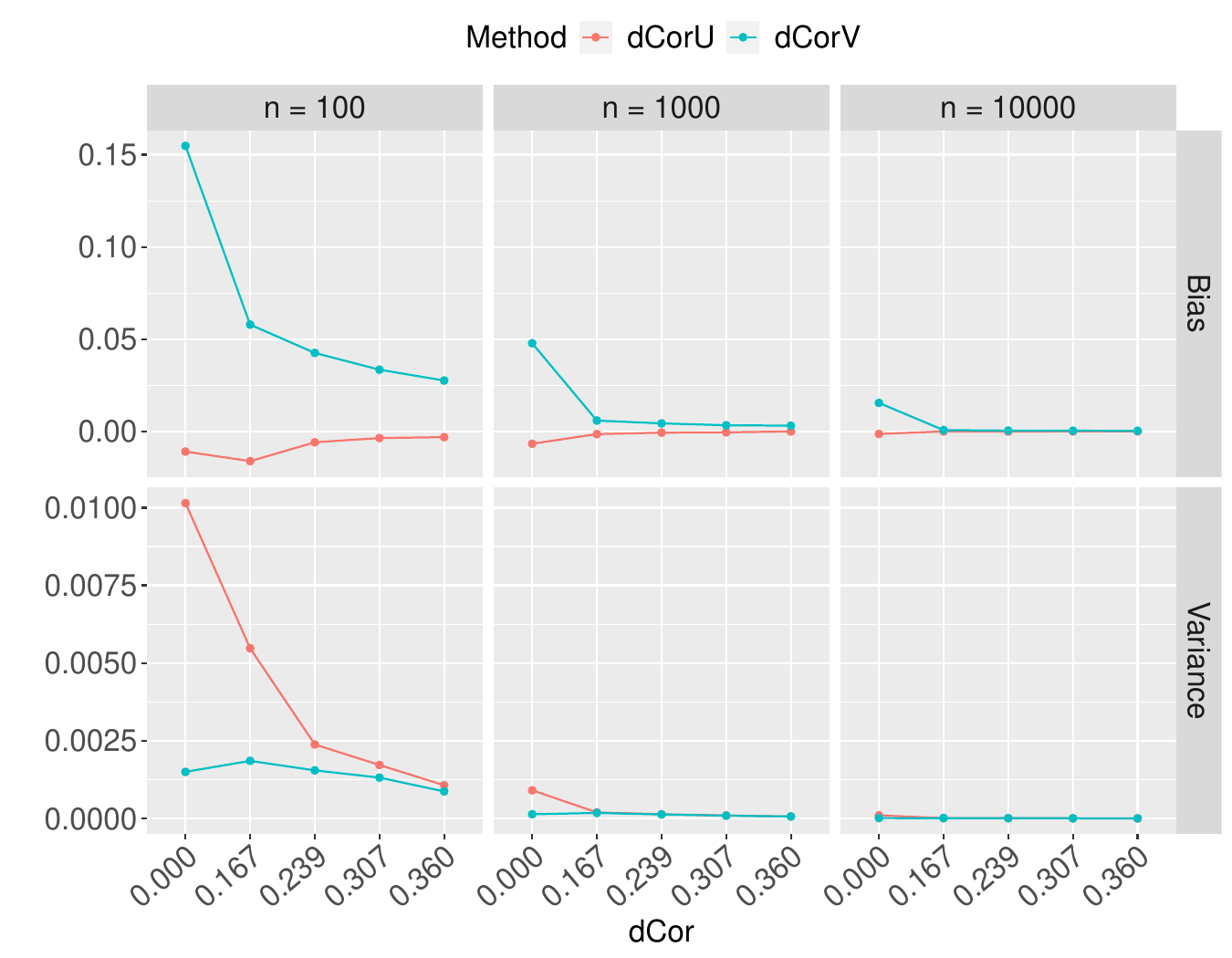}
			\caption{Bias and variance of dCorU and dCorV under the nonlinear model with different sample sizes ($n$) and distance correlation, from $\dcor = 0$ ($k = 0$) to $\dcor = 0.36$ ($k = 16$).}
			\label{fig:NLM_var}
		\end{figure}
		\newpage
		\subsection*{Convex linear combination}
		Figure \ref{fig:h} displays the MSE results for various bandwidth values ($h$ = 0.0025, 0.005, 0.01, 0.02, 0.04, 0.08, 0.16, 0.32) across different models and levels of dependence. The comparison was made through the three convex linear combinations: $\text{dCor}_{\hat{\lambda}_0}$, $\text{dCor(A)}_{\hat{\lambda}_0}$ and $\text{dCor(T)}_{\hat{\lambda}_0}$.\\
		\begin{figure}[htp]
			\centering
			\begin{subfigure}{1\textwidth}
				\includegraphics[width=\textwidth]{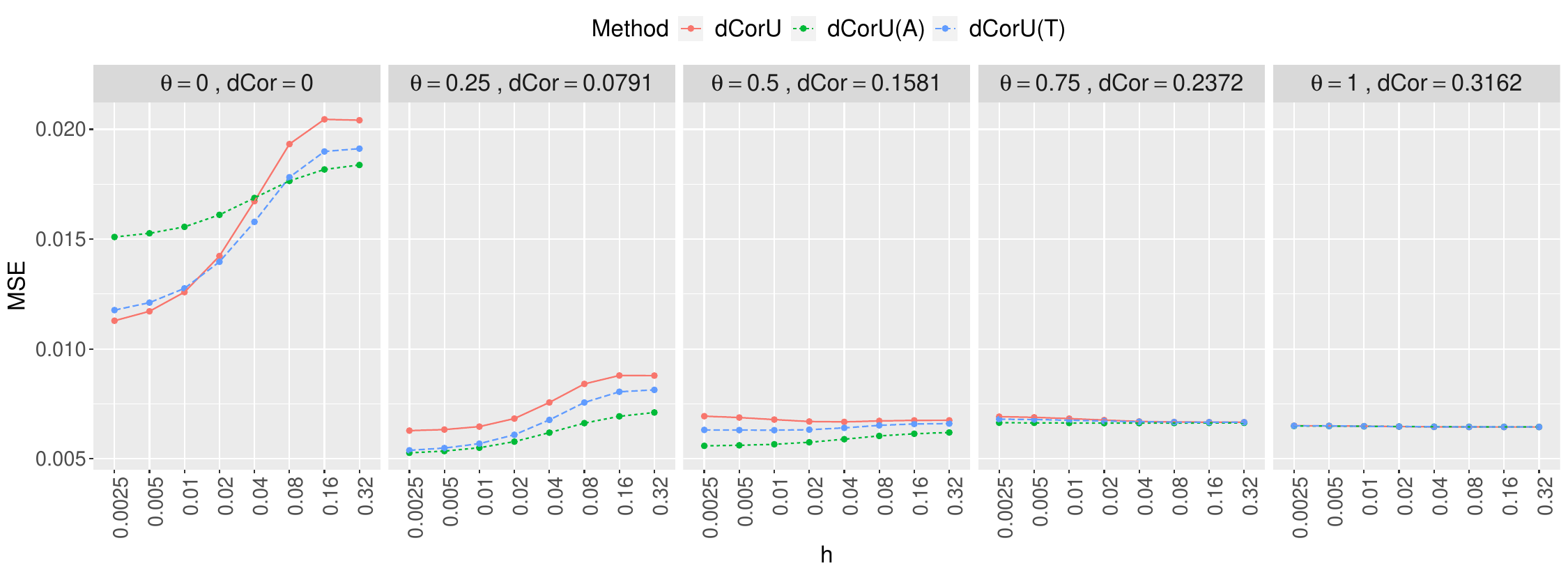}
				\caption{FGM-Model}
			\end{subfigure}
			\begin{subfigure}{1\textwidth}
				\includegraphics[width=1\textwidth]{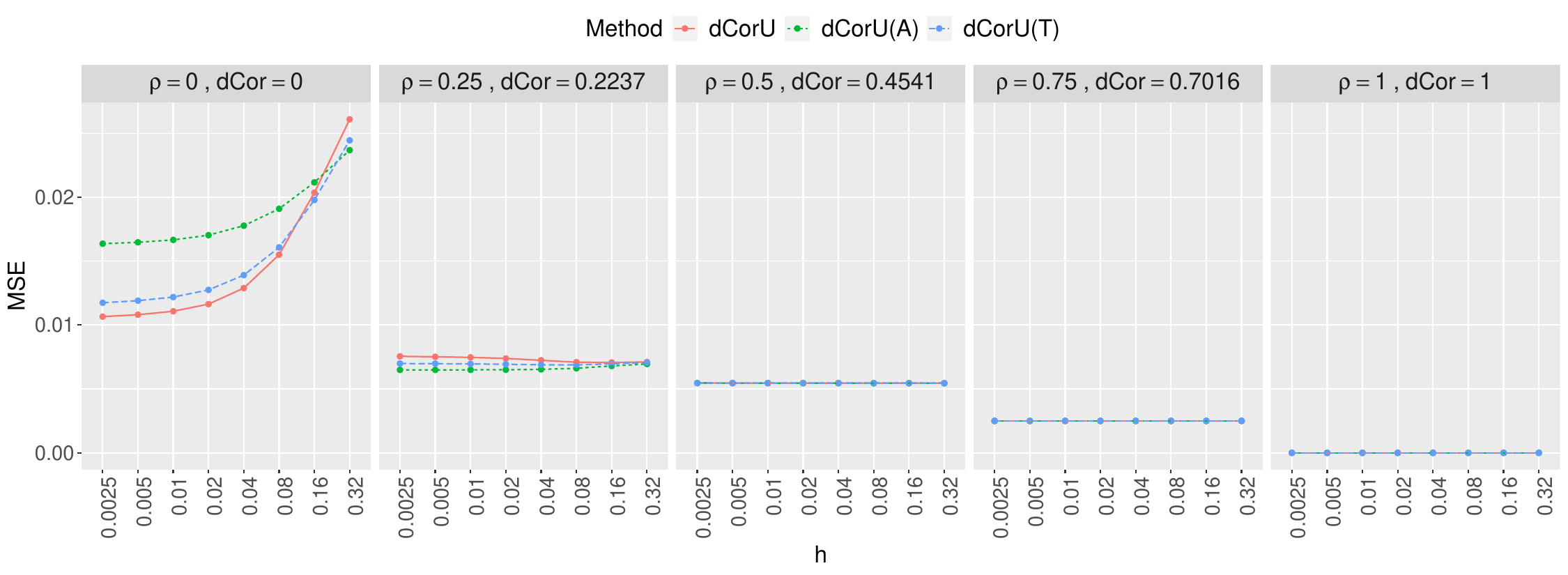}
				\caption{Bivariate normal model}
			\end{subfigure}
			\begin{subfigure}{1\textwidth}
				\includegraphics[width=\textwidth]{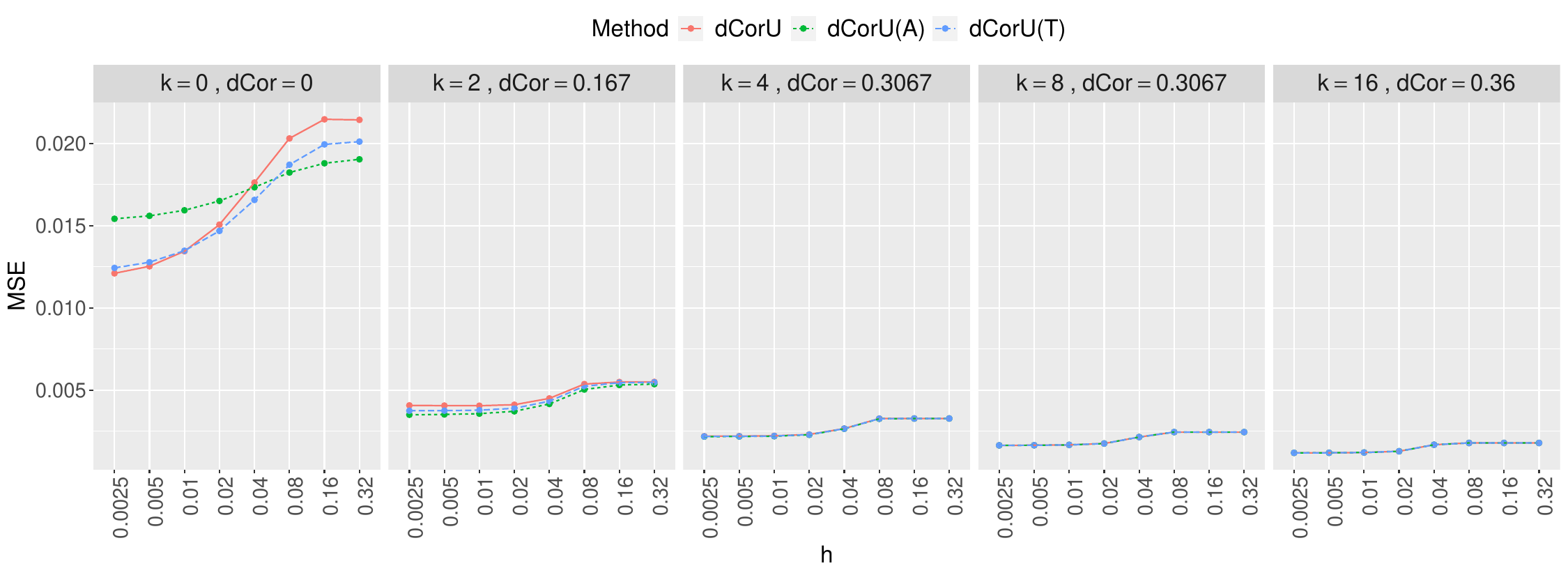}
				\caption{Nonlinear model}
			\end{subfigure}
			\caption{MSE values for the different bandwidths ($h$) across three models under different levels of dependence. The corresponding distance correlation values are also presented.}
			\label{fig:h}
		\end{figure}
		
		Figure \ref{fig:lambda_model} presents the MSE results for each model using different estimators, including the original estimators (dCorU, dCorV), alternatives to dCorU (dCorU(A), dCorU(T)), and their convex linear combinations. Specifically, $\text{dCor}_{\hat{\lambda}_0} = \hat{\lambda}_0$dCorU + $\left(1 - \hat{\lambda}_0\right)$dCorV with $\hat{\lambda}_0$ estimated using 1000 bootstrap replications. Similarly, $\text{dCor}_{\lambda_0} = \lambda_0$dCorU + $(1 - \lambda_0)$dCorV using the real optimal value of $\lambda_0$ computed with 1000 Monte Carlo repetitions. 
		\begin{figure}[htbp]
			\centering
			\includegraphics[width=1\textwidth]{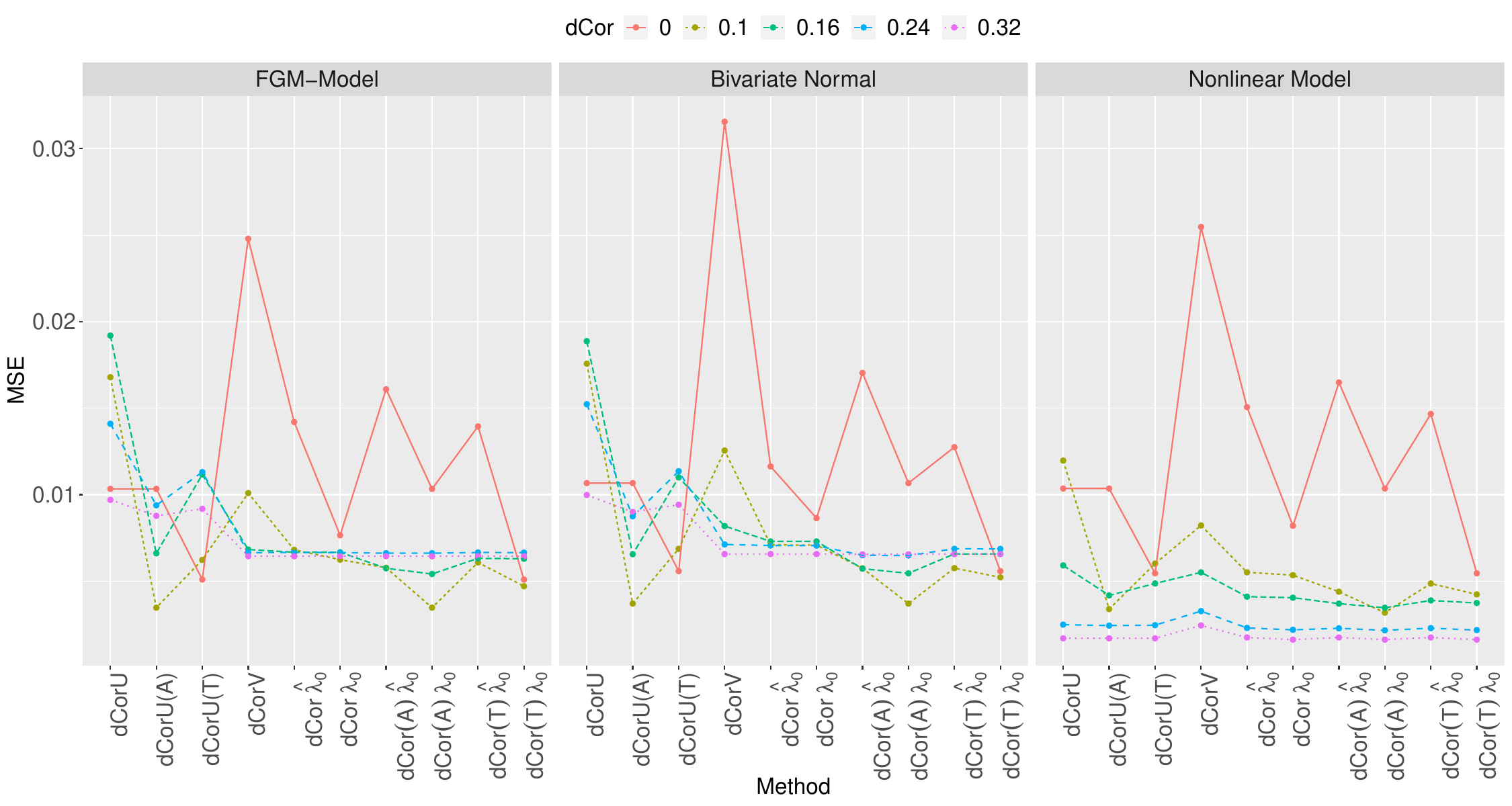}
			\caption{Comparison of the mean squared error (MSE) among the different estimators and their corresponding convex linear combinations for each model with a sample size of $n = 100$, along with the approximated distance correlation.} 
			\label{fig:lambda_model}
		\end{figure}
		
		\newpage
		\subsection*{Numerical results} 
		This appendix presents the complementary results of the simulation study. The tables show the mean, bias, variance, mean squared error (MSE) for each model with different parameters, using a sample size of $n$. \\ 
		
		For $n = 100$, the proposed dCorU(A) and dCorU(T) estimators are examined. Also the results for the convex linear combinations are provided. These tables include the mean, bias, variance and MSE for each combination, along with the bootstrap-estimated $\hat{\lambda}_0$ value. Since the real $\lambda_0$ can be obtained, the results for these combinations are also provided. This demonstrates the differences between the estimators of the distance correlation $\text{dCor}_{\lambda}$ with the bootstrap estimate $\hat{\lambda}_0$ and the real value $\lambda_0$.\\
		
		For $n = 1000$ and $n = 10000$, only results for the original estimators (dCorU, dCorV) are presented. Since the  negative values for these sample sizes are only observed under specific conditions and does not exhibit significant differences from the original estimators, then these results are not presented.\\ 
		\newpage
		\begin{table}[htbp!] 
			\begin{center}
				\begin{adjustbox}{width=1\textwidth}
					\begin{tabular}{c l  c  c  c c c c c c c c}
						\toprule
						\multicolumn{2}{c}{} & dCorU & dCorV & dCorU(A) & dCorU(T) &  ldCor & LdCor & ldCor(A) & LdCor(A) & ldCor(T) & LdCor(T) \\ 
						\cmidrule{2-12}
						\multicolumn{1}{c}{} & \multicolumn{1}{c}{$\hat{\lambda}_0 / \lambda_0$} & \multicolumn{2}{c}{} & \multicolumn{2}{ c }{} & \multicolumn{1}{ c}{$0.2737$} & \multicolumn{1}{ c}{$0.7167$} & \multicolumn{1}{ c}{$0.5065$} & \multicolumn{1}{ c}{$1$} & \multicolumn{1}{ c}{$0.3807$} &\multicolumn{1}{ c}{$1$} \\
						\cmidrule(rl){7-7} \cmidrule(rl){8-8} \cmidrule(rl){9-9} \cmidrule(rl){10-10} \cmidrule(rl){11-11} \cmidrule(rl){12-12}
						& Mean & -0.0191 & 0.1524 & 0.0923 & 0.0366 & 0.1055 & 0.0295 & 0.1220 & 0.0923 & 0.1083 & 0.0366 \\
						$\theta = 0$ & Bias & -0.0191 & 0.1524 & 0.0923 & 0.0366 & 0.1055 & 0.0295 & 0.1220 & 0.0923 & 0.1083 & 0.0366 \\ 
						$\dcor = 0$  & Var  & 0.0100  & 0.0016 & 0.0018 & 0.0038 & 0.0031 & 0.0068 & 0.0012 & 0.0018 & 0.0022 & 0.0038 \\
						& MSE  & 0.0103  & 0.0248 & 0.0103 & 0.0051 & 0.0142 & 0.0077 & 0.0161 & 0.0103 & 0.0140 & 0.0051 \\ 
						
						\bottomrule
						
						\multicolumn{1}{c}{} & \multicolumn{1}{c}{$\hat{\lambda}_0 / \lambda_0$} & \multicolumn{2}{c}{} & \multicolumn{2}{ c }{} & \multicolumn{1}{ c}{$0.2661$} & \multicolumn{1}{ c}{$0.5693$} & \multicolumn{1}{ c}{$0.4930$} & \multicolumn{1}{ c}{$1$} & \multicolumn{1}{ c}{$0.3705$} & \multicolumn{1}{ c}{$0.9538$} \\
						\cmidrule(rl){7-7} \cmidrule(rl){8-8} \cmidrule(rl){9-9}  \cmidrule(rl){10-10} \cmidrule(rl){11-11} \cmidrule(rl){12-12}
						& Mean & -0.0128 & 0.1546 & 0.0924 & 0.0398 & 0.1101 & 0.0593 & 0.1239 & 0.0924 & 0.1121 & 0.0451 \\
						$\theta = 0.1$      & Bias & -0.0444 & 0.1230 & 0.0608 & 0.0082 & 0.0785 & 0.0277 & 0.0923 & 0.0608 & 0.0805 & 0.0135 \\
						$\dcor = 0.032$     & Var  & 0.0103  & 0.0016 & 0.0020 & 0.0040 & 0.0032 & 0.0056 & 0.0013 & 0.0020 & 0.0023 & 0.0038 \\
						& MSE  & 0.0123  & 0.0168 & 0.0057 & 0.0041 & 0.0093 & 0.0064 & 0.0098 & 0.0057 & 0.0088 & 0.0040 \\ 
						\bottomrule
						
						\multicolumn{1}{c}{} & \multicolumn{1}{c}{$\hat{\lambda}_0 / \lambda_0$} & \multicolumn{2}{c}{} & \multicolumn{2}{ c }{} & \multicolumn{1}{ c}{$0.2312$} & \multicolumn{1}{ c}{$0.3766$} & \multicolumn{1}{ c}{$0.4327$} & \multicolumn{1}{ c}{$1$} & \multicolumn{1}{ c}{$0.3233$} & \multicolumn{1}{ c}{$0.6531$} \\
						\cmidrule(rl){7-7} \cmidrule(rl){8-8} \cmidrule(rl){9-9} \cmidrule(rl){10-10} \cmidrule(rl){11-11} \cmidrule(rl){12-12} 
						& Mean & 0.0175  & 0.1672 & 0.1018 & 0.0597  & 0.1326 & 0.1108 & 0.1389 & 0.1018 & 0.1324 & 0.0969 \\
						$\theta = 0.25$    & Bias & -0.0615 & 0.0881 & 0.0227 & -0.0194 & 0.0535 & 0.0318 & 0.0598 & 0.0227 & 0.0533 & 0.0179 \\
						$\dcor = 0.079$    & Var  & 0.0130  & 0.0023 & 0.0030 & 0.0059  & 0.0040 & 0.0524 & 0.0022 & 0.0030 & 0.0032 & 0.0044 \\
						& MSE  & 0.0168  & 0.0101 & 0.0035 & 0.0062  & 0.0068 & 0.0063 & 0.0058 & 0.0035 & 0.0061 & 0.0047 \\  
						\bottomrule
						
						\multicolumn{1}{c}{} & \multicolumn{1}{c}{$\hat{\lambda}_0 / \lambda_0$} & \multicolumn{2}{c}{} & \multicolumn{2}{ c }{} & \multicolumn{1}{ c}{$0.0844$} & \multicolumn{1}{ c}{$0.1006$} & \multicolumn{1}{ c}{$0.2692$} & \multicolumn{1}{ c}{$0.5212$} & \multicolumn{1}{ c}{$0.1973$} & \multicolumn{1}{ c}{$0.2479$} \\
						\cmidrule(rl){7-7} \cmidrule(rl){8-8} \cmidrule(rl){9-9} \cmidrule(rl){10-10} \cmidrule(rl){11-11} \cmidrule(rl){12-12}
						& Mean & 0.1040  & 0.2087 & 0.1438  & 0.1239  & 0.1999 & 0.1982 & 0.1912 & 0.1749 & 0.1920 & 0.0969 \\
						$\theta = 0.5$   & Bias & -0.0541 & 0.0506 & -0.0143 & -0.0342 & 0.0417 & 0.0400 & 0.0331 & 0.0167 & 0.0339 & 0.0179 \\ 
						$\dcor = 0.158$  & Var  & 0.0163  & 0.0043 & 0.0064  & 0.0100  & 0.0049 & 0.0051 & 0.0047 & 0.0051 & 0.0052 & 0.0044 \\
						& MSE  & 0.0192  & 0.0068 & 0.0066  & 0.0112  & 0.0067 & 0.0067 & 0.0058 & 0.0054 & 0.0063 & 0.0047 \\ 
						\bottomrule
						
						\multicolumn{1}{c}{} & \multicolumn{1}{c}{$\hat{\lambda}_0 / \lambda_0$} & \multicolumn{2}{c}{} & \multicolumn{2}{ c }{} & \multicolumn{1}{ c}{$ 0.0078$} & \multicolumn{1}{ c}{$0$} & \multicolumn{1}{ c}{$0.0850$} & \multicolumn{1}{ c}{$0.0970$} & \multicolumn{1}{ c}{$0.0158$}  & \multicolumn{1}{ c}{$0$} \\
						\cmidrule(rl){7-7} \cmidrule(rl){8-8} \cmidrule(rl){9-9} \cmidrule(rl){10-10} \cmidrule(rl){11-11} \cmidrule(rl){12-12}
						& Mean & 0.2049  & 0.2679 & 0.2148  & 0.2098  & 0.2674 & 0.2679 & 0.2634 & 0.2627 & 0.2670 & 0.2679 \\
						$\theta = 0.75$    & Bias & -0.0323 & 0.0307 & -0.0224 & -0.0274 & 0.0302 & 0.0307 & 0.0262 & 0.0255 & 0.0298 & 0.0307 \\
						$\dcor = 0.237$    & Var  & 0.0131  & 0.0057 & 0.0089  & 0.0106  & 0.0058 & 0.0057 & 0.0059 & 0.0060 & 0.0058 & 0.0057 \\
						& MSE  & 0.0141  & 0.0067 & 0.0094  & 0.0113  & 0.0067 & 0.0067 & 0.0066 & 0.0066 & 0.0067 & 0.0067 \\ 
						\bottomrule
						
						\multicolumn{1}{c}{} & \multicolumn{1}{c}{$\hat{\lambda}_0 / \lambda_0$} & \multicolumn{2}{c}{} & \multicolumn{2}{ c }{} & \multicolumn{1}{ c}{$0.0009$} & \multicolumn{1}{ c}{$0$} & \multicolumn{1}{ c}{$0.0055$} & \multicolumn{1}{ c}{$0$} & \multicolumn{1}{ c}{$0.0019$} & \multicolumn{1}{ c}{$0$} \\
						\cmidrule(rl){7-7} \cmidrule(rl){8-8} \cmidrule(rl){9-9} \cmidrule(rl){10-10} \cmidrule(rl){11-11} \cmidrule(rl){12-12}
						& Mean & 0.2966  & 0.3364 & 0.2980  & 0.2973  & 0.3364 & 0.3364 & 0.3362 & 0.3364 & 0.3363 & 0.3364 \\
						$\theta = 1$     & Bias & -0.0197 & 0.0202 & -0.0182 & -0.0189 & 0.0201 & 0.0202 & 0.0200 & 0.0202 & 0.0201 & 0.0202 \\
						$\dcor = 0.316$  & Var  & 0.0093  & 0.0060 & 0.0085  & 0.0088  & 0.0060 & 0.0060 & 0.0061 & 0.0060 & 0.0060 & 0.0060 \\
						& MSE  & 0.0097  & 0.0065 & 0.0088  & 0.0092  & 0.0065 & 0.0065 & 0.0065 & 0.0065 & 0.0065 & 0.0065 \\
						\bottomrule
					\end{tabular} 
				\end{adjustbox}
			\end{center}
			\caption{Results under the FGM-Model for different values of $\theta$ and $\dcor$. Specifically, ldCor denotes the combination $\text{dCor}_{\hat{\lambda}_0} = \hat{\lambda}_0$dCorU + $\left(1 - \hat{\lambda}_0\right)$dCorV with $\hat{\lambda}_0$ obtained through bootstrap. Similarly for dCorU(A) and dCorU(T). In the same way LdCor represents $\text{dCor}_{\lambda_0} = \lambda_0$dCorU + $\left(1 - \lambda_0\right)$dCorV, with $\lambda_0$ approximated using Monte Carlo, similarly for dCorU(A) and dCorU(T) with $n = 100$.}
		\end{table}
		\newpage
		\begin{table}[htbp!] 
			\begin{center}
				\begin{adjustbox}{width=1\textwidth}
					\small
					\begin{tabular}{c l  c  c  c c c c c c c c}
						\toprule
						\multicolumn{2}{c}{} & dCorU & dCorV & dCorU(A) & dCorU(T) &  ldCor &  LdCor & ldCor(A) &  LdCor(A)  & ldCor(T) &  LdCor(T) \\ 
						\cmidrule{2-12}
						\multicolumn{1}{c}{} & \multicolumn{1}{c}{$\hat{\lambda}_0 / \lambda_0$} & \multicolumn{2}{c}{} & \multicolumn{2}{ c }{} & \multicolumn{1}{ c}{$0.4923$} & \multicolumn{1}{ c}{$0.7707$} & \multicolumn{1}{ c}{$0.5982$} & \multicolumn{1}{ c}{$1$} & \multicolumn{1}{ c}{$0.5486$} & \multicolumn{1}{ c}{$1$} \\
						\cmidrule(rl){7-7} \cmidrule(rl){8-8} \cmidrule(rl){9-9} \cmidrule(rl){10-10} \cmidrule(rl){11-11} \cmidrule(rl){12-12}
						& Mean & -0.0102 & 0.1738 & 0.0936 & 0.0417 & 0.0832 & 0.0320 & 0.1258 & 0.0936 & 0.1014 & 0.0417 \\
						$\rho = 0$  & Bias & -0.0102 & 0.1738 & 0.0936 & 0.0417 & 0.0832 & 0.0320 & 0.1258 & 0.0936 & 0.1014 & 0.0417 \\
						$\dcor = 0$ & Var  & 0.0106  & 0.0013 & 0.0019 & 0.0038 & 0.0047 & 0.0076 & 0.0012 & 0.0019 & 0.0025 & 0.0038 \\
						& MSE  & 0.0107  & 0.0316 & 0.0107 & 0.0056 & 0.0116 & 0.0086 & 0.0170 & 0.0107 & 0.0128 & 0.0056 \\  
						\bottomrule
						
						\multicolumn{1}{c}{} & \multicolumn{1}{c}{$\hat{\lambda}_0 / \lambda_0$} & \multicolumn{2}{c}{} & \multicolumn{2}{ c }{} & \multicolumn{1}{ c}{$0.3968$} & \multicolumn{1}{ c}{$0.4193$} & \multicolumn{1}{ c}{$0.5359$} & \multicolumn{1}{ c}{$1$} & \multicolumn{1}{ c}{$0.4950$} & \multicolumn{1}{ c}{$0.6782$}\\
						\cmidrule(rl){7-7} \cmidrule(rl){8-8} \cmidrule(rl){9-9}  \cmidrule(rl){10-10} \cmidrule(rl){11-11} \cmidrule(rl){12-12}
						& Mean & 0.0310  & 0.1906 & 0.1088 & 0.0699  & 0.1273 & 0.1237 & 0.1468 & 0.1088 & 0.1309 & 0.1087  \\
						$\rho = 0.1$     & Bias & -0.0581 & 0.1014 & 0.0197 & -0.0192 & 0.0381 & 0.0346 & 0.0576 & 0.0197 & 0.0417 & 0.0196 \\
						$\dcor = 0.089$  & Var  & 0.0142  & 0.0026 & 0.0033 & 0.0065  & 0.0057 & 0.0059 & 0.0024 & 0.0033 & 0.0040 & 0.0048 \\
						& MSE  & 0.0176  & 0.0126 & 0.0037 & 0.0069  & 0.0071 & 0.0071 & 0.0058 & 0.0037 & 0.0058 & 0.0052 \\  
						\bottomrule
						
						\multicolumn{1}{c}{} & \multicolumn{1}{c}{$\hat{\lambda}_0 / \lambda_0$} & \multicolumn{2}{c}{} & \multicolumn{2}{ c }{} & \multicolumn{1}{ c}{$0.2228$} & \multicolumn{1}{ c}{$0.2169$} & \multicolumn{1}{ c}{$0.4191$} & \multicolumn{1}{ c}{$0.6113$} & \multicolumn{1}{ c}{$0.3914$} & \multicolumn{1}{ c}{$0.3764$} \\
						\cmidrule(rl){7-7} \cmidrule(rl){8-8} \cmidrule(rl){9-9} \cmidrule(rl){10-10} \cmidrule(rl){11-11} \cmidrule(rl){12-12}
						& Mean & 0.1099  & 0.2272 & 0.1482  & 0.1290  & 0.2011 & 0.2018 & 0.1941 & 0.1789 & 0.1888 & 0.1903 \\
						$\rho = 0.18$    & Bias & -0.0508 & 0.0665 & -0.0125 & -0.0317 & 0.0404 & 0.0411 & 0.0334 & 0.0182 & 0.0281 & 0.0296 \\
						$\dcor = 0.158$  & Var  &  0.0163 & 0.0038 & 0.0064  & 0.0100  & 0.0057 & 0.0056 & 0.0046 & 0.0051 & 0.0058 & 0.0057 \\
						& MSE  & 0.0189  & 0.0082 & 0.0066  & 0.0110  & 0.0073 & 0.0073 & 0.0057 & 0.0055 & 0.0066 & 0.0066 \\ 
						
						\bottomrule
						
						\multicolumn{1}{c}{} & \multicolumn{1}{c}{$\hat{\lambda}_0 / \lambda_0$} & \multicolumn{2}{c}{} & \multicolumn{2}{ c }{} & \multicolumn{1}{ c}{$0.0512$} & \multicolumn{1}{ c}{$0.0829$} & \multicolumn{1}{ c}{$0.3009$} & \multicolumn{1}{ c}{$0.3462$} & \multicolumn{1}{ c}{$0.1655$} & \multicolumn{1}{ c}{$0.1920$} \\
						\cmidrule(rl){7-7} \cmidrule(rl){8-8} \cmidrule(rl){9-9} \cmidrule(rl){10-10} \cmidrule(rl){11-11} \cmidrule(rl){12-12}
						& Mean & 0.1891  & 0.2710 & 0.2035  & 0.1963  & 0.2668 & 0.2642 & 0.2507 & 0.2477 & 0.2587 & 0.2567 \\
						$\rho = 0.25$    & Bias & -0.0346 & 0.0473 & -0.0202 & -0.0274 & 0.0431 & 0.0405 & 0.0270 & 0.0240 & 0.0350 & 0.0330 \\
						$\dcor = 0.224$  & Var  & 0.0140  & 0.0049 & 0.0084  & 0.0106  & 0.0052 & 0.0054 & 0.0058 & 0.0059 & 0.0057 & 0.0058 \\
						& MSE  & 0.0152  & 0.0071 & 0.0088  & 0.0114  & 0.0071 & 0.0071 & 0.0065 & 0.0065 & 0.0069 & 0.0069 \\  
						\bottomrule
						
						\multicolumn{1}{c}{} & \multicolumn{1}{c}{$\hat{\lambda}_0 / \lambda_0$} & \multicolumn{2}{c}{} & \multicolumn{2}{ c }{} & \multicolumn{1}{ c}{$0.0010$} & \multicolumn{1}{ c}{$0$} & \multicolumn{1}{ c}{$0.0660$} & \multicolumn{1}{ c}{$0.0658$} & \multicolumn{1}{ c}{$0.0157$} & \multicolumn{1}{ c}{$0.0190$} \\
						\cmidrule(rl){7-7} \cmidrule(rl){8-8} \cmidrule(rl){9-9} \cmidrule(rl){10-10} \cmidrule(rl){11-11} \cmidrule(rl){12-12}
						& Mean & 0.2948  & 0.3453 & 0.2964  & 0.2956  & 0.3452 & 0.3453 & 0.3401 & 0.3421 & 0.3445 & 0.3444 \\
						$\rho = 0.35$ 	 & Bias & -0.0198 & 0.0307 & -0.0183 & -0.0191 & 0.0306 & 0.0307 & 0.0274 & 0.0275 & 0.0299 & 0.0297 \\ 
						$\dcor = 0.315$  & Var  & 0.0096  & 0.0056 & 0.0087  & 0.0091  & 0.0056 & 0.0056 & 0.0058 & 0.0058 & 0.0057 & 0.0057 \\
						& MSE  & 0.0100  & 0.0066 & 0.0090  & 0.0094  & 0.0066 & 0.0066 & 0.0066 & 0.0066 & 0.0066 & 0.0066 \\ 
						\bottomrule
						
						\multicolumn{1}{c}{} & \multicolumn{1}{c}{$\hat{\lambda}_0 / \lambda_0$} & \multicolumn{2}{c}{} & \multicolumn{2}{ c }{} & \multicolumn{1}{ c}{$0$} & \multicolumn{1}{ c}{$0$} & \multicolumn{1}{ c}{$0$} & \multicolumn{1}{ c}{$0$} & \multicolumn{1}{ c}{$0$} & \multicolumn{1}{ c}{$0$} \\
						\cmidrule(rl){7-7} \cmidrule(rl){8-8} \cmidrule(rl){9-9} \cmidrule(rl){10-10} \cmidrule(rl){11-11} \cmidrule(rl){12-12}
						& Mean & 0.4420  & 0.4713 & 0.4420  & 0.4420  & 0.4713 & 0.4713 & 0.4713 & 0.4713 & 0.4713 & 0.4713 \\
						$\rho = 0.5$ 	 & Bias & -0.0122 & 0.0172 & -0.0122 & -0.0122 & 0.0172 & 0.0172 & 0.0172 & 0.0172 & 0.0172 & 0.0172 \\ 
						$\dcor = 0.454$  & Var  & 0.0065  & 0.0052 & 0.0065  & 0.0065  & 0.0052 & 0.0052 & 0.0052 & 0.0052 & 0.0052 & 0.0052 \\
						& MSE  & 0.0066  & 0.0055 & 0.0066  & 0.0066  & 0.0055 & 0.0055 & 0.0055 & 0.0055 & 0.0055 & 0.0055 \\ 
						\bottomrule
						
						\multicolumn{1}{c}{} & \multicolumn{1}{c}{$\hat{\lambda}_0 / \lambda_0$} & \multicolumn{2}{c}{} & \multicolumn{2}{ c }{} & \multicolumn{1}{ c}{$0$} & \multicolumn{1}{ c}{$0$} & \multicolumn{1}{ c}{$0$} & \multicolumn{1}{ c}{$0$} & \multicolumn{1}{ c}{$0$} & \multicolumn{1}{ c}{$0$} \\
						\cmidrule(rl){7-7} \cmidrule(rl){8-8} \cmidrule(rl){9-9} \cmidrule(rl){10-10} \cmidrule(rl){11-11} \cmidrule(rl){12-12}
						& Mean & 0.6951  & 0.7075 & 0.6951  & 0.6951  & 0.7075 & 0.7075 & 0.7075 & 0.7075 & 0.7075 & 0.7075 \\
						$\rho = 0.75$ 	 & Bias & -0.0066 & 0.0059 & -0.0066 & -0.0066 & 0.0059 & 0.0059 & 0.0059 & 0.0059 & 0.0059 & 0.0059 \\ 
						$\dcor = 0.702$  & Var  & 0.0028  & 0.0024 & 0.0028  & 0.0028  & 0.0024 & 0.0024 & 0.0024 & 0.0024 & 0.0024 & 0.0024 \\
						& MSE  & 0.0028  & 0.0025 & 0.0028  & 0.0028  & 0.0025 & 0.0025 & 0.0025 & 0.0025 & 0.0025 & 0.0025 \\ 
						\bottomrule
						
						\multicolumn{1}{c}{} & \multicolumn{1}{c}{$\hat{\lambda}_0 / \lambda_0$} & \multicolumn{2}{ c }{} & \multicolumn{1}{ c}{ } & \multicolumn{1}{ c}{ } & \multicolumn{1}{ c}{$0$} & \multicolumn{1}{ c}{$0$} & \multicolumn{1}{ c}{$0$} & \multicolumn{1}{ c}{$0$} & \multicolumn{1}{ c}{$0$} & \multicolumn{1}{ c}{$0$} \\
						\cmidrule(rl){7-7} \cmidrule(rl){8-8} \cmidrule(rl){9-9} \cmidrule(rl){10-10} \cmidrule(rl){11-11} \cmidrule(rl){12-12} 
						& Mean & 1.0000 & 1.0000 & 1.0000 & 1.0000 & 1.0000 & 1.0000 & 1.0000 & 1.0000 & 1.0000 & 1.0000 \\
						$\rho = 1$  & Bias & 0.0000 & 0.0000 & 0.0000 & 0.0000 & 0.0000 & 0.0000 & 0.0000 & 0.0000 & 0.0000 & 0.0000 \\
						$\dcor = 1$ & Var  & 0.0000 & 0.0000 & 0.0000 & 0.0000 & 0.0000 & 0.0000 & 0.0000 & 0.0000 & 0.0000 & 0.0000 \\
						& MSE  & 0.0000 & 0.0000 & 0.0000 & 0.0000 & 0.0000 & 0.0000 & 0.0000 & 0.0000 & 0.0000 & 0.0000 \\  
						\bottomrule
					\end{tabular} 
				\end{adjustbox}
			\end{center}
			\caption{Results under the bivariate normal model for different values of $\theta$ and $\dcor$. Specifically, ldCor denotes the combination $\text{dCor}_{\hat{\lambda}_0} = \hat{\lambda}_0$dCorU + $\left(1 - \hat{\lambda}_0\right)$dCorV with $\hat{\lambda}_0$ obtained through bootstrap. Similarly for dCorU(A) and dCorU(T). In the same way LdCor represents $\text{dCor}_{\lambda_0} = \lambda_0$dCorU + $\left(1 - \lambda_0\right)$dCorV,  with $\lambda_0$ approximated using Monte Carlo, similarly for dCorU(A) and dCorU(T) with $n = 100$.}
		\end{table}
		\newpage
		\begin{table}[htbp!] 
			\begin{center}
				\begin{adjustbox}{width=1\textwidth}
					\small
					\begin{tabular}{c l  c  c  c c c c c c c c}
						\toprule
						\multicolumn{2}{c}{} & dCorU & dCorV & dCorU(A) & dCorU(T) &  ldCor &  LdCor & ldCor(A) &  LdCor(A) & ldCor(T) &  LdCor(T )\\ 
						\cmidrule{2-12}
						\multicolumn{1}{c}{} & \multicolumn{1}{c}{$\hat{\lambda}_0 / \lambda_0$} & \multicolumn{2}{c}{} & \multicolumn{2}{ c }{} & \multicolumn{1}{ c}{$0.2734$} & \multicolumn{1}{ c}{$0.7391$} & \multicolumn{1}{ c}{$0.5017$} & \multicolumn{1}{ c}{$1$} & \multicolumn{1}{ c}{$0.3787$} & \multicolumn{1}{ c}{$1$}  \\
						\cmidrule(rl){7-7} \cmidrule(rl){8-8} \cmidrule(rl){9-9} \cmidrule(rl){10-10} \cmidrule(rl){11-11} \cmidrule(rl){12-12}
						& Mean & -0.0109 & 0.1548 & 0.0932 & 0.0411 & 0.1095 & 0.0323 & 0.1239 & 0.0932 & 0.1117 & 0.0411 \\
						$k = 0$ 	& Bias & -0.0109 & 0.1548 & 0.0932 & 0.0411 & 0.1095 & 0.0323 & 0.1239 & 0.0932 & 0.1117 & 0.0411 \\
						$\dcor = 0$ & Var  & 0.0102  & 0.0015 & 0.0017 & 0.0038 & 0.0031 & 0.0072 & 0.0011 & 0.0017 & 0.0021 & 0.0038 \\
						& MSE  & 0.01036 & 0.0255 & 0.0104 & 0.0055 & 0.0151 & 0.0082 & 0.0165 & 0.0104 & 0.0147 & 0.0055 \\ 
						\bottomrule
						
						\multicolumn{1}{c}{} & \multicolumn{1}{c}{$\hat{\lambda}_0 / \lambda_0$} & \multicolumn{2}{c}{} & \multicolumn{2}{ c }{} & \multicolumn{1}{ c}{$0.3020$} & \multicolumn{1}{ c}{$0.3971$} & \multicolumn{1}{ c}{$0.4242$} & \multicolumn{1}{ c}{$0.8322$} & \multicolumn{1}{ c}{$0.3621$} & \multicolumn{1}{ c}{$0.5994$} \\
						\cmidrule(rl){7-7} \cmidrule(rl){8-8} \cmidrule(rl){9-9} \cmidrule(rl){10-10} \cmidrule(rl){11-11} \cmidrule(rl){12-12}
						& Mean & 0.0648  & 0.1847 & 0.1056 & 0.0852  & 0.1484 & 0.1371 & 0.1511 & 0.1189 & 0.1487 & 0.1250 \\
						$k = 1$     	 & Bias & -0.0404 & 0.0796 & 0.0005 & -0.0199 & 0.0434 & 0.0320 & 0.0460 & 0.0138 & 0.0436 & 0.0200 \\ 
						$\dcor = 0.105$  & Var  & 0.0104  & 0.0019 & 0.0034 & 0.0056  & 0.0036 & 0.0043 & 0.0023 & 0.0030 & 0.0030 & 0.0038 \\
						& MSE  & 0.0120  & 0.0082 & 0.0034 & 0.0060  & 0.0055 & 0.0054 & 0.0044 & 0.0032 & 0.0049 & 0.0042 \\ 
						\bottomrule
						
						\multicolumn{1}{c}{} & \multicolumn{1}{c}{$\hat{\lambda}_0 / \lambda_0$} & \multicolumn{2}{c}{} & \multicolumn{2}{ c }{} & \multicolumn{1}{ c}{$0.3755$} & \multicolumn{1}{ c}{$0.4697$} & \multicolumn{1}{ c}{$0.4150$} & \multicolumn{1}{ c}{$0.6292$} & \multicolumn{1}{ c}{$0.3962$} & \multicolumn{1}{ c}{$0.5562$}  \\
						\cmidrule(rl){7-7} \cmidrule(rl){8-8} \cmidrule(rl){9-9} \cmidrule(rl){10-10} \cmidrule(rl){11-11} \cmidrule(rl){12-12}
						& Mean & 0.1524  & 0.2260 & 0.1576  & 0.1550  & 0.1984 & 0.1914 & 0.1976 & 0.1830 & 0.1979 & 0.1865 \\
						$k = 2$ 		 & Bias & -0.0147 & 0.0590 & -0.0095 & -0.0121 & 0.0313 & 0.0244 & 0.0306 & 0.0159 & 0.0308 & 0.0195 \\ 
						$\dcor = 0.167$  & Var  & 0.0057  & 0.0020 & 0.0041  & 0.0047  & 0.0031 & 0.0035 & 0.0028 & 0.0032 & 0.0029 & 0.0034 \\
						& MSE  & 0.0059  & 0.0055 & 0.0042  & 0.0049  & 0.0041 & 0.0041 & 0.0037 & 0.0035 & 0.0039 & 0.0037 \\ 
						\bottomrule
						
						\multicolumn{1}{c}{} & \multicolumn{1}{c}{$\hat{\lambda}_0 / \lambda_0$} & \multicolumn{2}{c}{} & \multicolumn{2}{ c }{} & \multicolumn{1}{ c}{$0.0512$} & \multicolumn{1}{ c}{$0.6556$} & \multicolumn{1}{ c}{$0.3009$} & \multicolumn{1}{ c}{$0.6687$}  & \multicolumn{1}{ c}{$0.1655$}  & \multicolumn{1}{ c}{$0.6630$} \\
						\cmidrule(rl){7-7} \cmidrule(rl){8-8} \cmidrule(rl){9-9} \cmidrule(rl){10-10} \cmidrule(rl){11-11} \cmidrule(rl){12-12}
						& Mean & 0.2318  & 0.2806 & 0.2318  & 0.2318  & 0.2587 & 0.2486 & 0.2587 & 0.2480 & 0.2587 & 0.2483 \\
						$k = 4$ 		 & Bias & -0.0071 & 0.0417 & -0.0070 & -0.0071 & 0.0198 & 0.0097 & 0.0198 & 0.0091 & 0.0198 & 0.0094 \\
						$\dcor = 0.239$  & Var  & 0.0024  & 0.0015 & 0.0024  & 0.0024  & 0.0019 & 0.0021 & 0.0019 & 0.0021 & 0.0019 & 0.0021 \\
						& MSE  & 0.0025  & 0.0033 & 0.0024  & 0.0025  & 0.0023 & 0.0022 & 0.0023 & 0.0022 & 0.0023 & 0.0022 \\  
						\bottomrule
						
						\multicolumn{1}{c}{} & \multicolumn{1}{c}{$\hat{\lambda}_0 / \lambda_0$} & \multicolumn{2}{c}{} & \multicolumn{2}{ c }{} & \multicolumn{1}{ c}{$0.4672$}  & \multicolumn{1}{ c}{$0.7673$} & \multicolumn{1}{ c}{$0.4672$}  & \multicolumn{1}{ c}{$0.7673$} & \multicolumn{1}{ c}{$0.4672$}  & \multicolumn{1}{ c}{$0.7673$} \\
						\cmidrule(rl){7-7} \cmidrule(rl){8-8} \cmidrule(rl){9-9} \cmidrule(rl){10-10} \cmidrule(rl){11-11} \cmidrule(rl){12-12}
						& Mean & 0.3039  & 0.3408 & 0.3039  & 0.3039  & 0.3235 & 0.3124 & 0.3235 & 0.3124 & 0.3235 & 0.3124 \\
						$k = 8$ 		 & Bias & -0.0029 & 0.0341 & -0.0029 & -0.0029 & 0.0168 & 0.0057 & 0.0168 & 0.0057 & 0.0168 & 0.0057 \\
						$\dcor = 0.307$  & Var  & 0.0017  & 0.0012 & 0.0017  & 0.0017  & 0.0015 & 0.0016 & 0.0015 & 0.0016 & 0.0015 & 0.0016 \\
						& MSE  & 0.0017  & 0.0025 & 0.0017  & 0.0017  & 0.0018 & 0.0016 & 0.0018 & 0.0016 & 0.0018 & 0.0016 \\ 
						\bottomrule 
						
						\multicolumn{1}{c}{} & \multicolumn{1}{c}{$\hat{\lambda}_0 / \lambda_0$} & \multicolumn{2}{c}{} & \multicolumn{2}{ c }{} & \multicolumn{1}{ c}{$0.4681$} & \multicolumn{1}{ c}{$0.7673$} & \multicolumn{1}{ c}{$0.4681$} & \multicolumn{1}{ c}{$0.7673$} & \multicolumn{1}{ c}{$0.4681$} & \multicolumn{1}{ c}{$0.7673$} \\
						\cmidrule(rl){7-7} \cmidrule(rl){8-8} \cmidrule(rl){9-9} \cmidrule(rl){10-10} \cmidrule(rl){11-11} \cmidrule(rl){12-12}
						& Mean & 0.3576  & 0.3883 & 0.3576  & 0.3576  & 0.3739 & 0.3637 & 0.3739 & 0.3637 & 0.3739 & 0.3637 \\
						$k = 16$ 		 & Bias & -0.0024 & 0.0283 & -0.0024 & -0.0024 & 0.0140 & 0.0038 & 0.0140 & 0.0038 & 0.0140 & 0.0038 \\
						$\dcor = 0.360$  & Var  & 0.0012  & 0.0010 & 0.0012  & 0.0012  & 0.0011 & 0.0012 & 0.0011 & 0.0012 & 0.0011 & 0.0012 \\
						& MSE  & 0.0012  & 0.0018 & 0.0012  & 0.0012  & 0.0013 & 0.0012 & 0.0013 & 0.0012 & 0.0013 & 0.0012 \\  
						\bottomrule
						
						\multicolumn{1}{c}{} & \multicolumn{1}{c}{$\hat{\lambda}_0 / \lambda_0$} & \multicolumn{2}{c}{} & \multicolumn{2}{ c }{} & \multicolumn{1}{ c}{$0.4055$} & \multicolumn{1}{ c}{$0.8397$} & \multicolumn{1}{ c}{$0.4055$} & \multicolumn{1}{ c}{$0.8397$} & \multicolumn{1}{ c}{$0.4055$} & \multicolumn{1}{ c}{$0.8397$} \\
						\cmidrule(rl){7-7} \cmidrule(rl){8-8} \cmidrule(rl){9-9} \cmidrule(rl){10-10} \cmidrule(rl){11-11} \cmidrule(rl){12-12}
						& Mean & 0.4122  & 0.4345 & 0.4122  & 0.4122  & 0.4255 & 0.4158 & 0.4255 & 0.4158 & 0.4255 & 0.4158  \\
						$k \to \infty$  & Bias & -0.0012 & 0.0210 & -0.0012 & -0.0012 & 0.0120 & 0.0023 & 0.0120 & 0.0023 & 0.0120 & 0.0023 \\
						$\dcor = 0.414$ & Var  & 0.0008  & 0.0007 & 0.0008  & 0.0008  & 0.0008 & 0.0008 & 0.0008 & 0.0008 & 0.0008 & 0.0008 \\
						& MSE  & 0.0008  & 0.0012 & 0.0008  & 0.0008  & 0.0009 & 0.0008 & 0.0009 & 0.0008 & 0.0009 & 0.0008 \\ 
						\bottomrule 
					\end{tabular} 
				\end{adjustbox}
			\end{center}
			\caption{Results under the nonlinear model for different values of $k$ and $\dcor$. Specifically, ldCor denotes the combination $\text{dCor}_{\hat{\lambda}_0} = \hat{\lambda}_0$dCorU + $\left(1 - \hat{\lambda}_0\right)$dCorV with $\hat{\lambda}_0$ obtained through bootstrap. Similarly for dCorU(A) and dCorU(T). In the same way LdCor represents $\text{dCor}_{\lambda_0} = \lambda_0$dCorU + $\left(1 - \lambda_0\right)$dCorV,  with $\lambda_0$ approximated using Monte Carlo, similarly for dCorU(A) and dCorU(T) with $n = 100$.}
		\end{table}
		\newpage
		\begin{table}[htbp!] 
			\begin{center}
				\begin{adjustbox}{width=1\textwidth}
					\small
					\begin{tabular}{c l  c  c  c c }
						\multicolumn{2}{c}{ } & \multicolumn{2}{c}{$n = 1000$} & \multicolumn{2}{c}{$n = 10000$}\\
						\toprule
						\multicolumn{2}{c}{} & dCorU & dCorV & dCorU & dCorV   \\ 
						\hline
						& Mean & -0.0064169473 & 0.0479105949 & -0.0011028850 & 0.0155715700 \\
						$\theta = 0$ & Bias & -0.0064159473 & 0.0479105949 & -0.001102885  & 0.0155715700 \\ 
						$\dcor = 0$  & Var  & 0.0009284552  & 0.0001375106 & 0.0001047187  & 0.0000169218 \\
						& MSE  & 0.0009696196  & 0.0024329357 & 0.0001058408  & 0.0002593804 \\ 
						\bottomrule
						
						
						& Mean & 0.0718297971  & 0.0907165625 & 0.0781968300  & 0.0797945300  \\
						$\theta = 0.25$ & Bias & -0.0072271444 & 0.0116596210 & -0.0008601117 & 0.0007375874 \\
						$\dcor = 0.079$ & Var  & 0.0013098307  & 0.0006638111 & 0.0001020316  & 0.0000977523  \\
						& MSE  & 0.0013620623  & 0.0007997579 & 0.0001026795  & 0.0000982084  \\  
						\bottomrule
						
						& Mean & 0.1565371293  & 0.1643796570 & 0.1575798101  & 0.1583507120 \\
						$\theta = 0.5$  & Bias & -0.0015767537 & 0.0062657740 & -0.0005340999 & 0.0002367772 \\
						$\dcor = 0.158$ & Var  & 0.0008823790  & 0.0007892589 & 0.0000966199  & 0.0000955250 \\
						& MSE  & 0.0008848652  & 0.0008285188 & 0.0000968182  & 0.0000955811 \\  
						\bottomrule
						
						& Mean & 0.2367121489  & 0.2416617704 & 0.2367525010  & 0.2372453021 \\
						$\theta = 0.75$ & Bias & -0.0004586757 & 0.0044909459 & -0.0004183631 & 0.0000744719 \\
						$\dcor = 0.237$ & Var  & 0.0008060114  & 0.0007675352 & 0.0000898021  & 0.0000893724 \\
						& MSE  & 0.0008062218  & 0.0007877038 & 0.00008989635 & 0.0000892975 \\  
						\bottomrule
						
						& Mean & 0.3162854420 & 0.3197722210 & 0.3158737320  & 0.3162222222 \\
						$\theta = 1$    & Bias & 0.0000576462 & 0.0035444680 & -0.0003540643 & -0.0000055621 \\
						$\dcor = 0.316$ & Var  & 0.0007237377 & 0.0007034444 & 0.0000805072  & 0.0000802804 \\
						& MSE  & 0.0007237410 & 0.0007160077 & 0.0000805601  & 0.0000802081 \\ 
						\bottomrule
						
					\end{tabular} 
				\end{adjustbox}
			\end{center}
			\caption{Results under the FGM-Model for different values of $\theta$ and $\dcor$ with $n = 1000, 10000$.}
		\end{table}
		\newpage
		\begin{table}[htbp!] 
			\begin{center}
				\begin{adjustbox}{width=1\textwidth}
					\small
					\begin{tabular}{c l  c  c  c c }
						\multicolumn{2}{c}{ } & \multicolumn{2}{c}{$n = 1000$} & \multicolumn{2}{c}{$n = 10000$}\\
						\toprule
						\multicolumn{2}{c}{} & dCorU & dCorV & dCorU & dCorV   \\ 
						\hline
						& Mean & -0.0053249323 & 0.0549054026 & -0.0020787712 & 0.0172301836 \\
						$\rho = 0$  & Bias & -0.0053249323 & 0.0549054026 & -0.0020787710 & 0.0017230182  \\ 
						$\dcor = 0$ & Var  & 0.0009774654  & 0.0001230894 & 0.0000938972  & 0.0000119813  \\
						& MSE  & 0.0010058203  & 0.0031376926 & 0.0000981340  & 0.0003088498  \\ 
						\bottomrule
						
						
						& Mean & 0.2214106074  & 0.2282662084 & 0.2232344402  & 0.2239136213  \\
						$\rho = 0.25$   & Bias & -0.0022892763 & 0.0045663247 & -0.0004654878 & 0.0002137027   \\
						$\dcor = 0.224$ & Var  & 0.0007963065  & 0.0007410121 & 0.0000778305  & 0.00007722296  \\
						& MSE  & 0.0008015473  & 0.0007618635 & 0.0000779771  & 0.00007726863  \\  
						\bottomrule
						
						& Mean & 0.4529598924  & 0.4558117978 & 0.4538119229  & 0.4540965165 \\
						$\rho = 0.5$    & Bias & -0.0011666120 & 0.0016852934 & -0.0003146241 & -0.0000300164 \\
						$\dcor = 0.454$ & Var  & 0.0005685453  & 0.0005568162 & 0.0000562083  & 0.0000560411 \\
						& MSE  & 0.0005699063  & 0.0005596564 & 0.0000562567  & 0.0000560420 \\  
						\bottomrule
						
						& Mean & 0.7009988519  & 0.7022233490 &  0.7014272263 & 0.7015496102 \\
						$\rho = 0.75$   & Bias & -0.0006204699 & 0.0006040272 & -0.0001920802 & -0.0000697245 \\
						$\dcor = 0.702$ & Var  & 0.0002495016  & 0.0002464804 & 0.0000245757  & 0.0000245234 \\
						& MSE  & 0.0002498865  & 0.0002468452 & 0.0000245905  & 0.0000245283 \\  
						\bottomrule
						
						& Mean & 1.0000000000 & 1.0000000000 & 1.0000000000 & 1.0000000000 \\
						$\rho = 1$  & Bias & 0.0000000000 & 0.0000000000 & 0.0000000000 & 0.0000000000 \\
						$\dcor = 1$ & Var  & 0.0000000000 & 0.0000000000 & 0.0000000000 & 0.0000000000 \\
						& MSE  & 0.0000000000 & 0.0000000000 & 0.0000000000 & 0.0000000000 \\ 
						\bottomrule
						
					\end{tabular} 
				\end{adjustbox}
			\end{center}
			\caption{Results under the bivariate normal model for different values of $\rho$ and $\dcor$ with $n = 1000, 10000$.}
		\end{table}
		\newpage
		\begin{table}[htbp!] 
			\begin{center}
				\begin{adjustbox}{width=1\textwidth}
					\small
					\begin{tabular}{c l  c  c  c c }
						\multicolumn{2}{c}{ } & \multicolumn{2}{c}{$n = 1000$} & \multicolumn{2}{c}{$n = 10000$}\\
						\toprule
						\multicolumn{2}{c}{} & dCorU & dCorV & dCorU & dCorV   \\ 
						\hline
						& Mean & -0.0066496747 & 0.0478197476 & -0.0013400740 & 0.01548219120 \\
						$k = 0$     & Bias & -0.0066496747 & 0.0478197476 & -0.001340074  & 0.0154821901  \\ 
						$\dcor = 0$ & Var  & 0.0009092382  & 0.0001394757 & 0.0001027066  & 0.00001687763  \\
						& MSE  & 0.0009534563  & 0.0024262040 & 0.0001044099  & 0.0002565608  \\ 
						\bottomrule
						
						
						& Mean & 0.1656174161  & 0.1729610014 & 0.1670618410 & 0.1677996200  \\
						$k = 2$         & Bias & -0.0014285839 & 0.0059150014 & 0.0000157970 & 0.0007536394 \\ 
						$\dcor = 0.167$ & Var  & 0.0001989045  & 0.0001817651 & 0.0000156266 & 0.0000154994  \\
						& MSE  & 0.0002009454  & 0.0002167524 & 0.0000156268 & 0.0000160534  \\ 
						\bottomrule
						
						& Mean & 0.2382283115  & 0.2432708574 & 0.2388520100  & 0.2393590100 \\
						$k = 4$         & Bias & -0.0006506885 & 0.0043918574 & -0.0000269810 & 0.0004799890 \\ 
						$\dcor = 0.239$ & Var  & 0.0001398824  & 0.0001335897 & 0.0000132147  & 0.0000131658 \\
						& MSE  & 0.0001403058  & 0.0001528781 & 0.0000132154  & 0.0000133844 \\ 
						\bottomrule
						
						& Mean & 0.3062829011  & 0.3101043201 & 0.3067652000 & 0.3071486210 \\
						$k = 8$         & Bias & -0.0004511185 & 0.0033703360 & 0.0000311962 & 0.0004145721 \\
						$\dcor = 0.307$ & Var  & 0.0000974818  & 0.0000948144 & 0.0000096070 & 0.0000095890 \\
						& MSE  & 0.0000976853  & 0.0001061736 & 0.0000096080 & 0.0000097520 \\  
						\bottomrule
						
						& Mean & 0.3599418001  & 0.3630928010 & 0.3599973100 & 0.3603135011 \\
						$k = 16$        & Bias & -0.0000112450 & 0.0031398350 & 0.0000442949 & 0.0003605299 \\
						$\dcor = 0.360$ & Var  & 0.0000677156  & 0.0000663919 & 0.0000068125 & 0.0000067988 \\
						& MSE  & 0.0000677158  & 0.0000762504 & 0.0000068083 & 0.0000069227 \\ 
						\bottomrule
						
					\end{tabular} 
				\end{adjustbox}
			\end{center}
			\caption{Results under the nonlinear model for different values of $k$ and $\dcor$ with $n = 1000, 10000$.}
		\end{table}

\end{document}